\documentclass[12pt]{article}
\usepackage{amsfonts,amsmath,amsthm, amssymb,a4}
\usepackage{graphicx}
\usepackage{caption}
\usepackage{setspace}
\usepackage{natbib}

\usepackage{floatflt}

\newcommand{\supp}{\mathrm{supp}\,}
\newcommand{\Prob}{\mathbb{P}}
\newcommand{\N}{\mathbb{N}}	
	
\newcommand{\R}{\mathbb{R}}	

\newcommand{\Or}{\mathcal{O}}

\newcommand{\X}{\mathbb{X}}	
\newcommand{\Y}{\mathbb{Y}}

\newcommand{\calL}{\mathcal{L}}

\newcommand{\Op}{\mathcal{F}}
\newcommand{\Ophat}{\widehat{\Op}}

\newcommand{\phihat}{\widehat{\varphi}}

\newcommand{\unknown}{\varphi}
\newcommand{\phidag}{\unknown^{\dagger}}

\newcommand{\Tdag}{T_\dag}
\newcommand{\That}{\widehat T}
\newcommand{\Tdaghat}{\widehat T_{n \dag}}

\newcommand{\Ex}{\mathbb{E}}
\newcommand{\Var}{\mathbb{Var}}

\newcommand{\argmin}{\mathop{\mathrm{argmin}}}

\newcommand{\dom}{\mathop{\mathrm{dom}}}

\def\Var{\mathop{\rm {\mathbb V}ar}\nolimits}%

\newcommand\independent{\protect\mathpalette{\protect\independenT}{\perp}}
\def\independenT#1#2{\mathrel{\rlap{$#1#2$}\mkern2mu{#1#2}}}

\theoremstyle{plain}
\newtheorem{theorem}{Theorem}
\newtheorem{corollary}{Corollary}
\newtheorem{lemma}{Lemma}

\theoremstyle{definition}
\newtheorem{definition}{Definition}
\newtheorem{assumption}{Assumption}
\newtheorem{example}{Example}

\theoremstyle{remark}

\begin{document}
\allowdisplaybreaks
\title{Adaptive estimation for some nonparametric instrumental variable models}
\author{Fabian Dunker\footnote{\noindent School of Mathematics and Statistics, University of Canterbury, Private Bag 4800, Christchurch 8140, New Zealand, \textit{Email:} fabian.dunker@canterbury.ac.nz}\\
University of Canterbury}

\date{}
\maketitle

\parindent0pt\textbf{Abstract:} 
The problem of endogeneity in statistics and econometrics is often handled by introducing instrumental variables (IV) which fulfill the mean independence assumption, i.e. the unobservable is mean independent of the instruments. When full independence of IV's and the unobservable is assumed, nonparametric IV regression models and nonparametric demand models lead to nonlinear integral equations with unknown integral kernels. We prove convergence rates for the mean integrated square error of the iteratively regularized Newton method applied to these problems. Compared to related results we derive stronger convergence results that rely on weaker nonlinearity restrictions. We demonstrate in numerical simulations  for a nonparametric IV regression that the method produces better results than the standard model.

\vskip .2in
\noindent {\sl MSC: AMS 2010 subject classification.} primary 62G08, secondary 62G20\\
\noindent {\sl Keywords and phrases:}  Nonparametric regression, instrumental variables, nonlinear inverse problems, regularization

\onehalfspacing
\parindent8pt

\section{Introduction}

Dependence of an unobservable error term and covariates is a frequent problem in statistical and econometrical modeling known as endogeneity. An efficient way to deal with
endogeneity is to use instrumental variables (IV) in the estimation. These are additional variables which can assumed to be independent or mean independent of the unobservable. In the context of nonparametric estimation the IV approach usually leads to ill-posed problems with an unknown operator that needs to be estimated. The solution $\varphi$ of the nonparametric IV problem can be characterized by a possibly nonlinear operator equation
\begin{align}\label{eq:opeq}
\Op (\varphi) = \psi.
\end{align}
In some regression models $\psi = 0$, in others $\psi$ is a function that has to be estimated from observations by some estimator $\widehat{\psi}$. The operator $\Op: \X \to \Y$
is an integral operator between some Banach or Hilbert spaces $\X$ and $\Y$ which is unknown in applications. Only an estimator $\Ophat$ is available.
The inverse of the operators $\Op$ or $\Ophat$ is usually not continuous. Even with an arbitrarily small variance in $\widehat\psi$ and $\widehat\Op$ we usually
have $\Var(\|\Ophat^{-1}\widehat\psi\|_\X) = \infty$ and the straightforward estimator $\widehat\varphi = \Ophat^{-1}\widehat\psi$ is typically inconsistent.
We discuss specific examples for nonparametric IV models and the related operators together with the respective literature in Section \ref{sec:application}.

In this paper we describe and analyze a consistent estimator for this type of problem, when $\Op$ is an operator between Hilbert spaces. The estimator is based on the iteratively regularized Gau{\ss}-Newton method (IRGNM) with iterated Tikhonov regularization defined below in \eqref{eq:it_tik}. Details about the method will be given in Section \ref{sec:estimation}.

This method was suggested by \cite{Baku:92}. Important monographs on this topic are \cite{BK:04b} and \cite{KNS:08}. These contributions consider only problems with known operators and deterministic right hand side in equation \eqref{eq:opeq}. The use of IRGNM for nonparametric IV problems was proposed and analyzed by \cite{DFHJM:14}. They derived rates for convergence in probability with a priori parameter choice using variational methods. 

The novelty of this paper is that we prove significantly faster convergence rates for the mean integrated squared error (MISE) rather than convergence in probability under a different set of assumptions. In addition, we propose adaptive estimation with Lepski\u\i's principle and prove rates for this case. 
Furthermore, we assume a significantly weaker nonlinearity condition for the operator $\Op$ which has a clear interpretation and is reasonable for most applications while the nonlinearity condition in \cite{DFHJM:14} is difficult to interpret and to check. We also prove faster rates of convergence when the regression function is smooth enough. Our proofs do not use variational methods. Instead we rely on spectral methods as in \cite{BauHohMun:09}. We also use a modification of Hoeffding's inequality from \cite{McD:89}. 

%

The paper is organized as follows. We discuss in Section \ref{sec:application} some IV models which fit into the framework of this paper and explain the estimator. The estimator is introduced in Section \ref{sec:estimation}. 
Section \ref{sec:convergence} contains convergence rate theorems. Finally, we present some numerical simulations in Section \ref{sec:numerics}. All proofs are in the Appendix.

\section{Nonparametric instrumental variable models}\label{sec:application}
In our general framework a function $\phidag$ is characterized by the possibly nonlinear operator equation
\begin{align}\label{eq:opeq0}
\Op(\phidag) = 0,
\end{align}
i.e. $\phidag$ is the true solution. Here $\Op :B_{2R}(\varphi_0) \subseteq \X \rightarrow \Y$ is an operator between Hilbert spaces with norms $\|\cdot\|_\X$ and $\|\cdot\|_\Y$ respectively. Tyical examples for $\X$ and $\Y$ are $L^2$ and $L^2$ based Sobolev spaces $H^i$ for $i=1,2,\ldots$. A ball $B_{2R}(\varphi_0)$ with radius $2R$ around an initial guess $\varphi_0$ is contained in the domain of $\Op$. In practice, large values of $R$ are possible. The operator equation is allowed to be ill-posed, i.e. $\Op^{-1}$ may not be continuous. Furthermore, the operator $\Op$ is not known in applications. Only a series of estimators $\Ophat_n :B_{2R}(\varphi_0) \subseteq \X \rightarrow \Y$ are available where $n$ denotes the sample size. We assume that $\phidag$ is a unique solution to \eqref{eq:opeq0} in $B_{2R}$, i.e. the problem is locally identified. In the following we discuss econometric examples for this setup.

\subsection{Nonparametric IV regression}

\paragraph{Mean independence}
The simplest nonparametric IV regression model has a separable error term and a mean independence condition
\begin{align}\label{eq:iv_reg_mean_independence}
Y = \varphi(X) + U \qquad \mbox{with } \Ex[U|Z] = 0.
\end{align}
Here and in all following models $Y$ and $U$ are univariate random variables, while $X$ and $Z$ can be multivariate and their dimensions do not have to coincide. The regressor $X$, the instrument $Z$ and the response $Y$ are observed, while the error term $U$ is unobservable. 

This model was proposed by \cite{NewPow:03} and \cite{florens:03}. It was further studied and applied in \cite{HalHor:05}, \cite{BluCheKri:07}, \cite{CheRei:11}, \cite{DFFR:11} \cite{FJV:11}, \cite{Hor:11}, \cite{JVV:11}, \cite{GS:12_reg}, \cite{ChePou:12}, \cite{Horowitz14}, \cite{ChenChrist:15}, \cite{BreJoh:15}, \cite{ChenChrist:18}, as well as \cite{Babii:20} among others. For an overview see \cite{Hor:14_survey}.

We can write \eqref{eq:iv_reg_mean_independence} equivalently as $\Ex[\varphi(X)|Z] = \Ex[Y|Z]$ and if the conditional densities $f_{X|Z}$ and $f_{Y|Z}$ exist, as
\begin{align}\label{eq:iv_reg_lin_op}
\int f_{X|Z}(x|z) \varphi(x) dx =  \int f_{Y|Z}(y|z)dy \qquad \mbox{for all } z \in \supp(Z).
\end{align}
We define the linear integral operator $(\Op_{ce}\varphi)(z) := \int f_{X|Z}(x|z) \varphi(x) dx$
with integral kernel $f_{X|Z}(x|z)$ and the function 
$\psi(z) := \int y f_{Y|Z}(y|z)dy$. Model \eqref{eq:iv_reg_mean_independence} can be given in operator form $(\Op_{ce}\varphi)(z) = \psi(z)$. 
The integral kernel $f_{X|Z}$ and thereby $\Op_{ce}$ as well as the function $\psi$ are unknown and have to be estimated from a sample of $Y,X,Z$. An Density estimators $\widehat f_{X|Z}$ and $\widehat f_{Y|Z}$ give estimators $\widehat\Op_{ce}$ and $\widehat \psi$ in a natural way.
While the main focus of this paper is on nonlinear operator equations, 
we use model \eqref{eq:iv_reg_mean_independence} 
as a benchmark for the IRGNM applied to model
\eqref{eq:iv_reg_full_independence} below.

The model identifies the regression function 
$\varphi$ if and only if $\Op_{ce}$ is injective. This property is called completeness, see \cite{Hault:11}, \cite{HF:15}, \cite{Andrews:17}, and \cite{BF:20}.

\paragraph{Full independence}\label{sec:full_ind}
In many applications the error term can be assumed to be independent of the instrument. Hence, mean independence of the instrument can be replace by full independence as proposed in \cite{DFHJM:14}
\begin{align}\label{eq:iv_reg_full_independence}
Y = \varphi(X) + U \qquad \mbox{with } U \independent Z \mbox{ and } \Ex[U] = 0.
\end{align}
Since the new assumptions $U \independent Z$ and $\Ex[U] = 0$ imply $\Ex[U|Z] = 0$ but not vice versa model \eqref{eq:iv_reg_full_independence} makes stronger assumptions than model \eqref{eq:iv_reg_mean_independence}.
Consequently, whenever \eqref{eq:iv_reg_mean_independence} identifies
the solution so does \eqref{eq:iv_reg_full_independence}. Furthermore, there are cases in which \eqref{eq:iv_reg_full_independence} can identify a solution, while
\eqref{eq:iv_reg_mean_independence} fails. This is for example the case with discrete instruments and continuous regressors as discussed in 
\cite{DFHJM:14}, \cite{Torgo}, \cite{HF:15}, \cite{CFF:19}, and \cite{Loh:19}.

We can translate model \eqref{eq:iv_reg_full_independence} into an operator equation by defining the operator 
\begin{align}\label{eq:nonlininstreg:opeq1}
(\widetilde\Op_{ind}(\varphi))(u,z) := \left(%
\begin{array}{c}
\Prob[Y - \varphi(X) \le u] - \Prob[Y - \varphi(X) \le u| Z=z] \\ 
\Ex[Y - \varphi(X)]%
\end{array}%
\right).
\end{align}
When $Y,X,Z$ have a joint density $f_{YXZ}$, taking the derivative with respect to $u$ yields the alternative operator
\begin{equation}\label{eq:nonlininstreg:opeq2}
(\Op_{ind}(\varphi))(u,z) := \left(\begin{array}{c}
\int f_{YXZ}(u+\varphi(x),x,z) - f_{YX}(u+\varphi(x),x) f_Z(z)\,dx\\
\int \varphi(x)f_X(x)\,dx-\int y f_Y(y)\,dy
\end{array}\right).
\end{equation}
Model \eqref{eq:iv_reg_full_independence} is equivalent to the operator equations $\widetilde\Op_{ind}(\varphi) = 0$ or $\Op_{ind}(\varphi) = 0$.
Note that the operators are nonlinear due to the first line of \eqref{eq:nonlininstreg:opeq1} or \eqref{eq:nonlininstreg:opeq2}. Furthermore, the operators are not known and have to be estimated. A density estimator 
$\widehat f_{YXZ}$ gives a straight forward estimator $\widehat \Op_{ind}$.

For any $\varphi$ that sets the first line of the operator $\Op_{ind}$ to $0$ also $c+ \varphi$ with $c\in\R$ sets it to $0$. In addition, for any $\varphi$, the second line of $\Op_{ind}$ is set to $0$ by $\varphi -\Ex[Y-\varphi(X)]$. Hence, for any solution $\varphi$ of the first line of the operator we have $\Op_{ind}\big(\varphi-\Ex[Y-\varphi(X)]\big)=0$. The nonlinear inverse problem is to find a $\varphi$ that solves the first line of $\Op_{ind}$. The second line is a parametric problem that can be estimated with the parametric rate.
When we discuss this example below we will only consider the first line of the operator as this is dominating the convergence rate. 

Let us denote the integral kernel of the first line of the operator \eqref{eq:nonlininstreg:opeq2} and its estimator by
$k_{ind}(y,x,z) := f_{YXZ}(y,x,z) - f_{YX}(y,x) f_Z(z)$ and $\widehat k_{ind}(y,x,z) := \widehat f_{YXZ}(y,x,z) - \widehat f_{YX}(y,x) \widehat f_Z(z)$
%
respectively. Then the first component of the operator reads $(\Op_{ind}(\varphi))(u,z) = \int k_{ind}(u+\varphi(x),x,z)dx$. 

\subsection{Quantile regression and non-separable models}

\paragraph{Nonparametric IV quantile regression}
Another model that leads to a different nonlinear operator equation
is nonparametric IV quantile regression proposed by \cite{HorLee:07}. For $q \in [0,1]$ the $q$-th quantile regression function $\varphi_q$ is characterized by
\begin{align}\label{eq:iv_quant_reg}
Y= \varphi_q(X) +U \qquad \Prob(U\leq 0|Z=z) = q \qquad \mbox{for all }z.
\end{align}
If the joint density $f_{YXZ}$ exists, the model is equivalent to an operator equation $\Op_q(\varphi_q)=0$ with
\begin{equation}\label{eq:iv_quant_reg_op}
(\Op_q(\varphi))(z) := \int F_{YXZ}(\varphi(x),x,z)\,dx- q f_Z(z)
\end{equation}
where $F_{YXZ}(y,x,z):=\int_{-\infty}^y f_{YXZ}(\tilde{y},x,z)\,d\tilde{y}$. Different estimation procedures for this model were proposed and analyzed in
\cite{HorLee:07}, \cite{ChePou:12}, \cite{GS:12_quant} \cite{DFHJM:14}, and \cite{Breunig:15}. Local identification properties of this and related models are discussed in \cite{CCLN:14}.

We can write $(\Op_q(\varphi))(z) = \int k_q(\varphi(x),x,z)dx$ with integral kernel
\[
k_q(y,x,z) := F_{YXZ}(y,x,z) - q f_{XZ}(x,z).
\]
Replacing $q f_Z(z)$ by $\int q f_{XZ}(x,z)dx$ is impractical in applications but makes it easier to discuss properties of \eqref{eq:iv_quant_reg_op} in this paper. $\Op_q$ and $k_q$ are unknown and have to be estimated. If we plug-in a density estimator $\widehat f_{YXZ}$, we get straight forward estimators $\widehat k_q$ and $\widehat{\Op}_q$.


\paragraph{Non-separable model}
A related example that falls in our framework is nonparametric IV regression with unseparable error, wich was proposed in
\cite{CheImbNew:07}. See also \cite{CheHan:05}. The model is
\begin{align}\label{eq:iv_nonsep_reg}
\begin{split}
Y=\phi(X,U) \qquad \mbox{with } &U \independent Z \mbox{ and}\\
&\phi(x,u) \mbox{ strictly monotonic increasing in } u.
\end{split}
\end{align}
It was pointed out in \cite{HorLee:07}, and \cite{CheImbNew:07} that this model is already
contained in model \eqref{eq:iv_quant_reg}. Let $F_U$ be the cumulative distribution function of $U$. Normalize $\widetilde U:= F_U(U)$ and 
$\widetilde\phi(x, \tilde u) := \phi(x,F_U^{-1}(\tilde u))$. Then $\widetilde U$ is uniformly distributed on $[0,1]$. The value of $\widetilde U$ corresponds to a quantile in model 
\eqref{eq:iv_quant_reg}. This reduces \eqref{eq:iv_nonsep_reg} to model \eqref{eq:iv_quant_reg} with $\varphi_q(x) = \widetilde\phi(x,q)$.

\paragraph{Further examples}
We briefly comment on further econometric models that fall into the framwork of this paper. A problem that has a similar mathematical structure as IV regression appears in some nonparametric demand models for differentiated products. It was considered with mean independence assumption similar to \eqref{eq:iv_reg_mean_independence} in \cite{BH:11}, \cite{BH:14} and with full independence similar to \eqref{eq:iv_reg_full_independence} in \cite{DHK:14}. Some models for games of incomplete information lead to a nonlinear inverse problem with deterministic operator, see for example \cite{FS:10}. Nonlinear inverse problems with deterministic operators also occur in functional linear quantile regression (without instrumental variables) as in \cite{Kato:12}. The estimator in this paper can be applied to these type of problems. However, the error analysis would be different since there is no randomness in the operator.
Also related are nonparametric ARCH($\infty$) models which can be treated as linear inverse problem, see \cite{LM:05}. Further linear inverse problems in econometrics are discussed in \cite{CFR:07}.


\section{Estimation}\label{sec:estimation}

\subsection{The estimator}


Remember that $\phidag$ denotes the true solution and let $\varphi_0$ be an initial guess.
Our method is based on linearizing $\Op$ which motivates the following assumption.

\begin{assumption}\label{ass:dif}
\begin{enumerate}
\item $\|\phidag - \varphi_0\|_\X < R$
\item $\Op$ and all $\Ophat_n$ are Fr\'echet differentiable on $B_{2R}(\varphi_0)$ with Fr\'echet derivatives $\Op'$ and $\Ophat_n'$ respectively.
\end{enumerate}
\end{assumption}

The iteratively regularized Gau\ss-Newton method with iterated Tikhonov regularization consists of two nested iterations. The outer iteration is a Newton method. It starts at $\varphi_0$ and produces in the $j$-th step the estimate $\phihat_{j+1}$. In the $j$-th step the operator is linearized as $\Ophat(\varphi) \approx \Ophat_n'[\phihat_j](\varphi - \phihat_j) + \Ophat_n(\phihat_j)$. A regular Newton method would invert the linear operator $\Ophat_n'[\phihat_j]$ to compute the next step. Due to the ill-posedness, this would be unstable and we use a regularized inverse instead. The regularized inverse is computed by $m$-times iterated Tikhonov regularization which is the inner iteration of the method. In the following scheme the Newton iteration is indexed by $j$ and the Tikhonov iteration by $i$, and $\alpha_j > 0$ is a regularization parameter
\begin{equation}\label{eq:it_tik}
\begin{alignedat}{1}
&\text{for } j = 0 \text{ to } J\\
&\quad \overline\varphi_{j+1,0} :=\, \phihat_j\\
&\quad\text{for } i = 0 \text{ to } m-1\\
&\quad\quad\overline\varphi_{j+1,i+1} := \argmin\limits_{\varphi \in \X} \left(\|\Ophat_n'[\phihat_j](\varphi - \phihat_j) + \Ophat_n(\phihat_j)\|_{\Y}^2 + \alpha_j \|\varphi - \overline\varphi_{j+1,i}\|_\X^2 \right)\\
&\quad \text{end}\\
&\quad\widehat \varphi_{j+1} :=\, \overline\varphi_{j+1,m}\\
&\quad\mbox{stop if } \|\phihat_j-\varphi_0 \|_\X > 2R \mbox{ and set } \phihat_{j+1} = \varphi_0\\
&\text{end}.
\end{alignedat} 
\end{equation}
As usual for Newton methods, convergence can fail if the initial guess $\widehat\varphi_0$ is too far from the true solution $\varphi^\dag$. In practice and in simulations the method proves to be quite robust to the choice of $\widehat\varphi_0$. If no a priori information about $\varphi^\dag$ is available, $\widehat{\varphi}_0 = 0$ is usually a good choice.

With a small $\alpha_j$ the method has a large variance due to the ill-posedness of $\Op$. While a larger $\alpha_j$ controls the variance but adds some bias. We choose $\alpha_0$ large enough to stabilize the problem and let $\alpha_j$ decay in every Newton step by 
\begin{align}\label{eq:alpha}
\alpha_{j+1} = q_\alpha \alpha_j \quad \mbox{with some fixed} \quad 0<q_\alpha<1
\end{align}
to reduce the bias. A second parameter that has to be chosen is the number of inner iterations $m$. A large $m$ is of advantage  for very smooth $\phidag$. We will address the choice of $\alpha_0$ and $m$ in Section \ref{sec:spec_source} and Assumption \ref{ass:saturation}. The Newton iteration needs to be stopped at an appropriate iteration step. The size of the regularization parameter is linked to the number of steps. Hence, the number of steps corresponds to a bias variance trade-off. We will investigate parameter choice with a priori knowledge in Section \ref{sec:rates} and fully data driven in Section \ref{sec:Lepski}.

We introduce the following notations for shorter formulas
\[
\Tdag := \Op'[\phidag] \qquad \That_{n,j} := \Ophat_n'[\phihat_j] \qquad \Tdaghat := \Ophat'_n[\phidag].
\]
An alternative formulation of the method can be obtained by using the functional calculus. Let $\That_{n,j}^*$ denote the adjoint operator of $\That_{n,j}$
and set 
\begin{equation}\label{eq:g_alpha}
g_\alpha(\lambda) := \frac {(\lambda + \alpha)^m - \alpha^m}{\lambda(\lambda + \alpha)^m}~. 
\end{equation}
Then 
\eqref{eq:it_tik} is equivalent to
\begin{equation}\label{eq:irgnm_spektral}
\begin{split}
&\text{for } j = 0 \text{ to } J\\
&\quad\phihat_{j+1} = \varphi_0 + g_{\alpha_j}(\That_{n,j}^* \That_{n,j})\That_{n,j}^* \left(\That_{n,j}(\phihat_j-\varphi_0) - \Ophat_n(\phihat_j) \right)\\
&\quad\mbox{stop if } \|\phihat_j-\varphi_0 \|_\X > 2R \mbox{ and set } \phihat_{j+1} = \varphi_0\\
&\text{end}.
\end{split}
\end{equation}

\begin{example}[Fr\'echet differentiability]\label{ex:frechet}
Assumption \ref{ass:dif} is usually fulfilled in our examples. The operators are well defined and Fr\'echet differentiable on the whole space under mild integrability conditions on the joint density $f_{YXZ}$. The Fr\'echet derivative of the operator in \eqref{eq:nonlininstreg:opeq2} exists when $f_{YXZ}$ is partially differentiable in the first variable. The operator in \eqref{eq:iv_quant_reg_op} is differentiable without further assumptions.
\begin{align*}
(\Op'_{ind}[\varphi]\psi)(u,z) &= \left(\begin{array}{c}
\int \big[ \frac{\partial}{\partial y}f_{YXZ}(u+\varphi(x),x,z) - \frac{\partial}{\partial y}f_{YX}(u+\varphi(x),x) f_Z(z)\big] \psi(x)\,dx\\
\int \psi(x) f_X(x)\,dx
\end{array}\right),\\
(\Op'_q[\varphi](\psi))(z) &= \int f_{YXZ}(\varphi(x),x,z)\psi(x)\,dx.
\end{align*}
Note that the derivatives are linear integral operators with kernel $\frac{\partial}{\partial y}k(\varphi(x),x,z)$.
\end{example}

\section{Convergence Rates}\label{sec:convergence}

The convergence theory is presented in four steps. We start by introducing assumptions for the general operator equation \eqref{eq:opeq0} as well as for the IV regression models \eqref{eq:nonlininstreg:opeq2} and \eqref{eq:iv_quant_reg_op}. Afterwards, we state convergence rate result for the MISE with a priori choice of the stopping parameter $j$. Then we compare the result to \cite{HorLee:07}. The last step is a theorem with data driven choice of $j$ by Lepski\u\i's principle.


\subsection{Assumptions}

\subsubsection{Smoothness condition}\label{sec:spec_source}

As usual for nonparametric methods a smoothness assumption has to be imposed on the true solution $\phidag$ to get convergence rates. In our setup with an ill-posed operator equation \eqref{eq:opeq} it is necessary to link the smoothness of $\phidag$ to the smoothing properties of the operator $\Op$. An efficient and popular way to formulate this is a source condition. The following definition uses the functional calculus.
\begin{definition}\label{def:spec_sc}
Let $\Lambda : [0 , \infty) \; \to [0 , \infty)$ be continuous, strictly increasing with $\Lambda(0) = 0$. A representation of the initial error as
\begin{equation}\label{eq:spec_sc}
\varphi_0 - \phidag = \Lambda(T^*_\dag T_\dag)\omega\;, \qquad \omega \in \X \text{ with } \rho:=\|\omega\|_\X
\end{equation}
is called a spectral source condition and $\Lambda$ is called an index function.
\end{definition}
When $\Tdag$ is a linear integral operator with kernel $\frac{\partial}{\partial y}k(\phidag(x),x,z)$ as in Example \ref{ex:frechet}, this definition can be interpreted in the following way. We assume for simplicity that $\Tdag$ is compact which is for example the case if $\frac{\partial}{\partial y}k(\phidag(x),x,z)$ is continuous. It was shown in \cite{reade841} and \cite{reade842} that the singular values of such an operator decay at least polynomially if $\frac{\partial}{\partial y}k(\phidag(x),x,z)$ belongs to a Sobolev space, and exponentially if $\frac{\partial}{\partial y}k(\phidag(x),x,z)$ is analytic. 

Let $(\sigma_t,\, u_t,\, v_t)$ be a singular system for $\Tdag$. The source condition \eqref{eq:spec_sc} implies for $e_0 = \varphi_0 - \phidag$
\[
\omega = \sum_{t\in \N}\frac{\langle e_0,\, v_t \rangle}{\Lambda(\sigma_t^2)} u_t \in \X \qquad \mbox{and thereby} \qquad
\sum_{t=1}^\infty \left(\frac{\langle e_0,\, v_t \rangle}{\Lambda(\sigma_t^2)}\right)^2 < \infty.
\]
Hence, a $\omega$ fulfilling \eqref{eq:spec_sc} only exists if $\Lambda$ compensates the decay of the singular values in a way that $\langle e_0,\, v_t \rangle\Lambda(\sigma_t^2)^{-1}$ is square summable. The decay of singular values describes the smoothing properties of the $\Tdag$ with respect to the singular vectors. While the decay of $\langle e_0,\, v_t \rangle$ describes the smoothness of $e_0$ with respect to the singular vectors. Thus, the rate of decay for $\Lambda(x)$ when $x \searrow 0$ compares these two degrees of smoothness. For the examples above the source condition compares the smoothness of $f_{YXZ}$ with the smoothness of the regression function $\phidag$.

When $\sigma_t$ and $\langle e_0,\, v_t \rangle$ both decay polynomially or both decay exponentially, i.e. $\sigma_t \lesssim \exp(-c_\sigma t)$ and $\langle e_0,\, v_t \rangle \lesssim \exp(-c_{e_0} t)$ with some constants $c_\sigma$ and $c_{e_0}$, the source condition is fulfilled with $\Lambda(x) = x^\mu$. Where $\mu >0$ is a sufficiently small constant. A source condition with polynomial $\Lambda$ is called a H\"older source condition, which is a concept that goes back to \cite{lavrentev62} and \cite{morozov68}. 
For exponential decay of $\sigma_t$ but only polynomial decay of $\langle e_0,\, v_t \rangle$ the source condition holds when the operator is rescaled to $\|\Tdag\| < 1$ and $\Lambda(x) = (-\ln(x))^{-p}$ with some $0<p$. 
This choice of $\Lambda$ was proposed by \cite{mair94} and \cite{hohage:97} and 
is called logarithmic source condition.

%
%

Despite the word ``condition'' in the name ``source condition'' it is rather a relation that selects an index function. Corollary 2 in \cite{mh08} shows that for any compact injective operator $\Tdag$ and any $e_0$ exists an index function $\Lambda$ such that a source condition is fulfilled.

In this paper we focus on H\"older source conditions with $\mu > 1/2$. Notice that this implies $e_0 \in \text{Range}(\Op'[\phidag]^*)$. The case of $\mu \le 1/2$ and logarithmic source conditions was analyzed in \cite{DFHJM:14}. We make the formal assumption:
\begin{assumption}\label{ass:source}
The true solution $\phidag$ fulfills a source condition \eqref{eq:spec_sc} with sufficiently small $\rho$ and with an index function that satisfies $\Lambda(x) = \mathcal{O}(x^\mu)$ for $x \searrow 0$ with $\mu > 1/2$.
\end{assumption}
\begin{example}
For nonparametric IV regression with full independence \eqref{eq:nonlininstreg:opeq2} Assumption \ref{ass:source} implies that $\varphi_0-\phidag$ is in the same smoothness class as
$\frac{\partial}{\partial y}f_{YXZ}(u+\varphi(x),x,z) - \frac{\partial}{\partial y}f_{YX}(u+\varphi(x),x) f_Z(z).$
For nonparametric IV quantile regression \eqref{eq:iv_quant_reg_op} Assumption \ref{ass:source} implies that $\varphi_0-\phidag$ is in the same smoothness class as
$f_{YXZ}(\phidag(x),x,z).$
The smoother $\varphi_0-\phidag$ the larger is $\mu$.
\end{example}

Closely related to the smoothness of the true solution is the choice of the parameters $\alpha_0$ and $m$ for the IRGNM. In the following assumption $\|T\|_{\calL(\X,\X)} := \sup_{\varphi} \{\|T\varphi\|_\X~| ~\|\varphi\|_\X=1\}$ denotes the usual operator norm for linear operators.
\begin{assumption}\label{ass:saturation}
\begin{enumerate}
\item The number of iterations of the Tikhonov regularization $m$ is larger or equal to $\mu$ in the source conditions $m\ge \mu$, i.e. $\Lambda(x)^{-1}x^{m} = \mathcal{O}(1)$ for $x \searrow 0$.

\item The initial regularization parameter $\alpha_0$ is large enough such that $\alpha_0 \geq \|\Tdaghat^* \Tdaghat\|_{\calL(\X,\X)}/(1-q_\alpha)$.
\end{enumerate}
\end{assumption}
Both parameters need to be large enough but there is not much harm in choosing them larger than necessary. Any $\alpha_0$ and $m$ fulfilling Assumption \ref{ass:saturation} will lead to comparable estimates. However, increasing $\alpha_0$ will lead to a few more Newton steps. The lower bound of $\alpha_0$ depends on the derivative of the estimated operator and is thereby random but not unknown. 

As usual for nonparametric methods the rate of convergence increases if the true solution $\phidag$ is smoother, i.e. if $\mu$ is larger. But this increase is only realized if $m \ge \mu$. However, $m$ does not act as a regularization parameter. 
Since the inner iteration is numerically cheap it is save to chose a larger value for $m$ without having a significant disadvantage.

\subsubsection{Nonlinearity restriction}
The non-linearity of $\Op$ needs to be restricted for the algorithm to work. We use a Lipschitz condition on the derivative for this purpose.
\begin{assumption}\label{ass:lipschitz}
There exists $L > 0$ such that 
\begin{equation}\label{eq:lipschitz}
\|\Ophat_n'[\xi_1] - \Ophat_n'[\xi_2]\|_{\calL(\X,\Y)} \leq L \|\xi_1 - \xi_2\|_\X 
\end{equation}
almost surely for all $\xi_1 , \xi_2 \in B_R(\phidag)$ and large $n$.
\end{assumption}
The special structure of the of $\Op_{ind}$ and $\Op_q$ allows us to replace Assumption \ref{ass:lipschitz} for the IV regression examples by the following alternative. Lemma \ref{lem:nonlin_restrict} in the appendix shows that Assumption \ref{ass:lipschitz_alt} implies Assumption \ref{ass:lipschitz}.
\begin{assumption}\label{ass:lipschitz_alt}
For the operators \eqref{eq:nonlininstreg:opeq2} and \eqref{eq:iv_quant_reg_op} the integral kernels $k_{ind}
$, $k_q$, and their estimates are twice differentiable with respect to $y$ with bounded derivative and the support of the instrument has finite measure $\mu(\supp(Z)) < \infty$. Furthermore, the integral kernels are estimated by an estimator which is strongly consistent for the second derivative. There exists $L>0$ such that
\[
\mu(\supp(Z))\sup_{y, z, w} \left| \frac{\partial^2}{\partial y^2} k(y, z, w) \right| < L
\] 
with $k = k_{ind}$ or $k=k_q$ respectively.
\end{assumption}

Common nonparametric density estimators are strongly consistent. Assumption \ref{ass:lipschitz_alt} implies for the operators \eqref{eq:nonlininstreg:opeq2} and \eqref{eq:iv_quant_reg_op}
\[
\sup_{y, x, z} \left| \frac{\partial^2}{\partial y^2} f_{YXZ}(y, x, z) \right| < \infty \qquad \mbox{or} \qquad \sup_{y, x, z} \left| \frac{\partial}{\partial y} f_{YXZ}(y, x, z) \right| < \infty
\]
respectively.

\subsubsection{Concentration inequalities}
The estimation error in the operator and its derivative needs to be bounded by exponential inequalities.
\begin{assumption}\label{ass:concentration}
There are constants $c_1, c_2, c_3, c_4 \ge 0$ such that for all $n \in \N$ and all $\tau \ge 0$
\begin{align}\label{eq:concentration_noi}
&\Prob \left\{ \left|\|\Ophat_n(\phidag)\|_\Y - \Ex\|\Ophat_n(\phidag)\|_\Y \right| \ge \sqrt{\tau \Var \left(\|\Ophat_n(\phidag)\|_\Y\right)} \right\} \leq c_1 e^{-c_2\tau}\quad \mbox{and}\\
\begin{split}\label{eq:concentration_der}
&\Prob \Bigg\{\left|\|\Tdaghat - \Tdag\|_D^{1+\mu} - \Ex\left(\|\Tdaghat - \Tdag\|_D^{1+\mu}\right) \right| \geq\\
& \hspace{5.3cm}\sqrt{ \tau \Var\left(\|\Tdaghat - \Tdag\|_D^{1+\mu}\right)}  \Bigg\} \leq c_3 e^{-c_4\tau}.
\end{split}
\end{align}
Where $\|\cdot\|_D$ is the operator norm $\|\cdot\|_{\mathcal{L}(\X,\Y)}$ or some norm that dominates the operator norm. 
\end{assumption}
The following lemma shows that Assumption \ref{ass:concentration} holds for the IV regression applications \eqref{eq:nonlininstreg:opeq2} and \eqref{eq:iv_quant_reg_op} under mild conditions when $\Y$ is a $L^2$ space and $\|\cdot\|_D$ is the Hilbert-Schmidt norm. The Hilbert-Schmidt norm bounds the operator norm from above and is denoted by $\|\cdot\|_{HS}$. For linear integral operators it coincides with the $L^2$ norm of the integral kernel.

\begin{lemma}\label{lem:concentration} 
Consider the operators \eqref{eq:nonlininstreg:opeq2} and \eqref{eq:iv_quant_reg_op} as maps into $L^2(U,Z)$ or $L^2(Z)$ respectively. Assume that $f_{YXZ}$ is estimated by a kernel density estimator with a product kernel composed of a one-dimensional kernel $K_Y$ and two multivariate kernels $K_X$ and $K_Z$ corresponding to the dimensions $\dim(X)=d_X$ and $\dim(Z)=d_Z$ with joint bandwidth $h$. Assume for \eqref{eq:nonlininstreg:opeq2} that $n^{-1}h^{-d_z-1} = O(1)$, and $n^{-1-2\mu}h^{-(1+\mu)(d_x+d_z+3)} = O(1)$. Assume for \eqref{eq:iv_quant_reg_op} that $n^{-1}h^{-d_z} = O(1)$, and $n^{-1-2\mu}h^{-(1+\mu)(d_x+d_z+2)} = O(1)$. Then constants $c_2,c_4>0$ exist such that for all $\tau \ge 0$ 
\begin{align}\label{eq:concentration_op}
\Prob \left\{ \left| \|\Ophat_n(\phidag)\|_{L^2} - \Ex\|\Ophat_n(\phidag)\|_{L^2} \right| \ge \sqrt{\tau \Var \left(\|\Ophat_n(\phidag)\|_{L^2}\right)} 
\right\} \leq 2 e^{-c_2\tau} 
\end{align}
and
\begin{align}\label{eq:concentration_de}
\Prob \Bigg\{\left|\|\Tdaghat - \Tdag\|_{HS}^{1+\mu} - \Ex\left(\|\Tdaghat - \Tdag\|_{HS}^{1+\mu}\right) \right| \geq \sqrt{ \tau \Var\left(\|\Tdaghat - \Tdag\|_{HS}^{1+\mu}\right)}  \Bigg\} \leq 2 e^{-c_4\tau}.
\end{align}
\end{lemma}

\subsection{Convergence rates with a priori parameter choice}\label{sec:rates}
Our first convergence rate theorem assumes that $\mu$ in Assumption \ref{ass:source} is known, i.e. the smoothness of the true solution is known. Adaptive estimation will be discussed in the next section.

\begin{theorem}\label{theo:spec_conv}
Let Assumptions \ref{ass:dif}, \ref{ass:source}, \ref{ass:saturation}, \ref{ass:lipschitz}, and \ref{ass:concentration} hold. Define the stopping index as
\[
J:= \argmin\limits_{j\in\N} \left(\alpha_j^\mu + \alpha_j^{-1/2} \Ex\big[\|\Ophat_n(\phidag)\|_\Y^2\big] \right)
\]
and set
\begin{align}\label{eq:J*}
J^*:=\begin{cases}
        J \hspace{1cm} \text{ if } \phihat_j \in B_{2R}(\varphi_0) \text{ for } j=1,\ldots,J\\
	0 \hspace{1cm} \text{ else}.
       \end{cases}
\end{align}

Then,
\[
\Ex\left[\|\phihat_{J^*} -\phidag\|_\X^2\right] = \Or \left(\left(\Ex\big[\|\Ophat_n(\phidag)\|_\Y^2\big]\right)^{\frac{2\mu}{2\mu + 1}} + \Ex\big[\|\Tdaghat - \Tdag\|_{\calL(\X,\Y)}^{2+2\mu}\big] \right).
\]
\end{theorem}

In the special cases of the IV regression examples in $L^2$ spaces, the convergence rate can be given more explicitly.
\begin{corollary}\label{cor:rates_npiv}
Let $\X$ be an $L^2$ space and let Assumptions \ref{ass:dif}, \ref{ass:source}, \ref{ass:saturation}, \ref{ass:lipschitz_alt} and the conditions of Lemma \ref{lem:concentration} hold. Assume that in the case of operator \eqref{eq:nonlininstreg:opeq2} the density $f_{YXZ}$ and that in case of operator \eqref{eq:iv_quant_reg_op} the function $F_{YXZ}$ is $r$ times differentiable and is estimated with a kernel estimator where the kernel is of order at least $r$. If $J^*$ is chosen as in Theorem \ref{theo:spec_conv}, then
\[
\Ex\big[\|\phihat_{J^*} -\phidag\|_{L^2}^2\big] = \Or \left((n^{-1} h^{-(d_Z+1)})^{\frac{2\mu}{2\mu + 1}} + n^{-1-\mu} h^{-((1+2\mu)(d_X+d_Z+2)+1)} + h^{\frac{4\mu r}{2\mu + 1}} \right).
\]
\end{corollary}

\subsection{Comparison to an alternative quantile regression estimator}
We can compare our result to the rates for nonparametric quantile regression in \cite{HorLee:07}. They use nonlinear Tikhonov regularization for the operator $\Op_q$ and proved optimal rates under assumptions which are more restrictive than ours. In contrast to our rates, their rates do not depend on the derivative $\Tdaghat$. The main challenge for nonlinear Tikhonov regularization
\[
\phihat = \argmin_\varphi\|\Op_q \varphi \|_\Y^2 + \alpha\|\varphi\|_\X^2 
\] 
is to find the minimizer of the nonlinear functional $\|\Op_q \varphi \|_\Y^2 + \alpha\|\varphi\|_\X^2$ which usually has multiple local minima. In \cite{HorLee:07} it is assumed that this minimum is known exactly which is unrealistic in practice. A convergence analysis which takes the performance of a minimization algorithm into account would typically lead to a different rate which also depends on some derivative depending on the particular minimization algorithm. The IRGNM does not have this problem since we only have to solve a linear least squares problem in every Newton step. For a fair comparison of the convergence rates, we will assume that the term $\left(\Ex\big[\|\Ophat_n(\phidag)\|_\Y^2\big]\right)^{\frac{2\mu}{2\mu + 1}}$ dominates our rate. For sake of simplicity we assume $d_z=1$. Hence, our rate in the case of $\X=L^2$ is
\[
\Ex\big[\|\phihat_{J^*} -\phidag\|_{L^2}^2\big] = \Or \left((n^{-1} h^{-2} + h^{2r} )^{\frac{2\mu}{2\mu + 1}}\right),
\]
while the rate in \cite{HorLee:07} is $\Or\left(n^{-(2\beta-1)/(2\beta+a)}\right)$ with $a$ and $\beta$ defined an their paper as $\alpha$ and $\beta$. \cite{HorLee:07} restricts the bandwidth choice to 
\[
h = C_h n^{-\gamma} \quad \text{ with }\quad \frac{2\beta+a -1)}{2r(2\beta+a} <\gamma< \frac{a}{2(2\beta+a)}.
\]
Under their assumptions on $a$ and $\beta$, the density estimator achieves the rate 
\[
n^{-1} h^{-2} + h^{2r} = \Or\left(n^{-\frac{2\beta +a-1}{2\beta-1}}\right).
\]
Note that this will only coincide with the optimal rate $n^{-\frac{2r}{2r+2}}$ in special cases. Futhermore, in their notation the source condition in our Assumption \ref{ass:source} holds for any $\mu < (\beta-\frac{1}{2})/a$. Hence, $2\mu/(2\mu+1) < (2\beta-1)/(2\beta +a -1)$, where $\mu$ can be chosen such that the left hand side is arbitrarily close to the right hand side. Therefore, under their assumptions our rate is arbitrarily close to
\[
\left(n^{-\frac{2\beta +a-1}{2\beta-1}}\right)^\frac{2\beta-1}{2\beta +a -1} = n^{-\frac{2\beta-1}{2\beta+a}},
\] 
which is also the optimal rate in Theorem 2 and 3 in \cite{HorLee:07}.

Our assumptions are less restrictive which leads to faster rates of convergence in many cases compared to \cite{HorLee:07}. Firstly, we have no restrictions on $h$ which means we can chose the optimal bandwidth $h = \Or\left(n^{-\frac{1}{2r+2}}\right)$ and achieve 
\[
\Ex\big[\|\phihat_{J^*} -\phidag\|_{L^2}^2\big] = \Or \left(n^{-\frac{2r}{2r+2}\frac{2\mu}{2\mu + 1}}\right)
\]
This has also the advantage that we can choose $h$ by standard data driven bandwidth selectors like cross-validation. It is pointed out in \cite{HorLee:07} that a bandwidth selector for their assumptions does not yet exist.

Secondly, we have no upper bound on $\mu$ while their method requires $\mu \le 1$. Hence, our rate can be significantly better for smooth $\phidag$. This is not a contradiction to their optimality result, as their result only holds under their more restrictive assumptions.

\subsection{Convergence rates for adaptive estimation }\label{sec:Lepski}

The parameter choice \eqref{eq:J*} in Theorem \ref{theo:spec_conv} and Corollary \ref{cor:rates_npiv} depends on the unknown $\mu$
, which is unfeasible in practice. We present in this section convergence rates for a data driven choice of the stopping parameter $J$ by Lepski{\u\i}'s principle. This is a popular parameter choice rule in the context of statistical inverse problems, see \cite{Tsybakov:00}, \cite{BauHoh:05}, \cite{Mathe:06}, \cite{BauHohMun:09}, and \cite{HW:16}. In the context of nonparametric IV, Lepski{\u\i}'s principle was used for adaptive estimation in \cite{ChenChrist:15}. 
The following theorem gives convergence rates of the MISE with a Lepski{\u\i} type parameter choice. We lose a logarithmic factor compared to Theorem \ref{theo:spec_conv}. The constant $C_d$ and $\gamma_{nl}$ used in the theorem are specified in the appendix in formula \eqref{eq:der_noi_est} and in Lemma \ref{lem:nonlin} respectively.

\begin{theorem}\label{theo:lepskii_risk}
Let the assumptions of Theorem \ref{theo:spec_conv} hold. 
For all sequences $\delta^{noi}_n$, $\sigma^{noi}_n$, $\delta^{der}_n$, and $\sigma^{der}_n$ such that
\begin{align*}
&\delta^{noi}_n \geq \Ex(\|\Ophat_n(\phidag)\|_\Y), &&(\sigma^{noi}_n)^2 \geq \Var(\|\Ophat_n(\phidag)\|_\Y),\\
&\delta^{der}_n \geq \Ex(\|\Tdaghat - \Tdag\|_{\calL(\X,\Y)}^{1+\mu}), &&(\sigma^{der}_n)^2 \geq \Var(\|\Tdaghat - \Tdag\|_{\calL(\X,\Y)}^{1+\mu})
\end{align*}
set
\[
\widetilde\Phi^{noi}_n(j) := \sqrt{\frac{m}{\alpha_j}} \left(\delta^{noi}_n + \ln((\sigma^{noi}_n)^{-2}) \sigma^{noi}_n \right) + C_d \rho (\delta^{der}_n + \ln((\sigma^{der}_n)^{-2}) \sigma^{der}_n)
\]
and define the Lepski{\u\i} stopping parameter by
\[
J_{Lep}:= \min \left\{j \leq J_{max} \Big| \|\phihat_i - \phihat_j\|_\X \leq 4(1+\gamma_{nl}) \widetilde\Phi^{noi}_n(j) \quad \text{for all } i = 1,\, \ldots,\, J_{max} \right\}
\]
and the stopping parameter by
\[
J^*:=\begin{cases}
        J_{Lep} &\text{if } \phihat_j \in B_{2R}(\varphi_0) \text{ for } j=1,\ldots,J_{max}\\
	0 \hspace{1cm} &\text{else}.
       \end{cases}
\]
Then
\begin{align*}
E\big[\|&\phihat_{J^*} -\phidag\|_\X^2\big]\\
& = \Or \left(\left[(\delta^{noi}_n)^2 + \ln\big((\sigma^{noi}_n)^{-1}\big)(\sigma^{noi}_n)^2 \right]^{\frac{2\mu}{2\mu + 1}} + (\delta^{der}_n)^2 + \ln\big((\sigma^{der}_n)^{-1}\big)(\sigma^{der}_n)^2 \right).
\end{align*}
\end{theorem}

\section{Numerical examples}\label{sec:numerics}

We evaluate the small sample behavior of the estimator based on the IRGNM with simulated data. As a test problem we use a nonparametric IV regression consistent with models \eqref{eq:iv_reg_mean_independence} and \eqref{eq:iv_reg_full_independence}, with univariate covariate $X$ and instrument $Z$. This setup allows to compare our estimator with an estimator which solves \eqref{eq:iv_reg_lin_op} with iterated Tikhonov regularization.

\subsection{Implementation}
The test problem described in the next section is solved on the domain 
\[
\supp(Y) \times \supp(X) \times \supp(Z) = [-1/2 , 1/2] \times [0 , 1] \times [0 , 1]
\]
discretized by an equidistant grid with $100 \times 100 \times 100$ nodes. The joint density is estimated on this grid by a standard adaptive density estimator. Trying different density estimators, we found that both method are quite robust with respect to the density estimate. They tolerate some undersmoothing of the density as long as the stopping index $J$ and the regularization parameter $\alpha$ are chosen properly. In the simulations below the same density estimate is used for both the IRGNM and the iterated Tikhonov regularization which allows for a fair comparison of the methods. The initial guess for both methods is the constant function with the value $\Ex[Y]$, and the penalty functional for both methods is the squared $H^1$ norm.

The Fr\'{e}chet derivative is implemented as in Example \ref{ex:frechet}. The partial derivative of the density and the derivatives for the $H^1$ norm are computed by the central differencing scheme. Operators and norms are evaluated using numerical integration. We tried rectangle rule, trapezoid rule and Simpson's rule but could not find a significant difference in the output of the estimator.

The least squares problems in each step of the iterated Tikhonov regularization and in the inner iteration of the IRGNM are computed by QR decomposition. Note that only one QR decomposition is needed in every Newton step. We tried different numbers of iterations $m$ in the inner iteration of the IRGNM and the iterated Tikhonov regularization for the test example below. No significant difference in the results was observed, which indicates that $\mu$ is not large.

The regularization parameters for the example below are $\alpha_0 = 1$ and $\alpha_{n+1} = 0.9 \alpha_n$. Lepski{\u\i}'s principle is used to find the stopping parameter of the Newton iteration. For the alternative estimator using model \eqref{eq:iv_reg_mean_independence} the regularization parameter $\alpha$ has to be chosen instead, which is done by Lepski{\u\i}'s principle as well. The iterated Tikhonov regularization is computed for a large number of different $\alpha$. Then one of these approximation is chosen by Lepski{\u\i}'s principle. 
Hence, both methods are fully data driven.

\subsection{Simulations}
The regressor of the test example is generated by some function $g$ and a random variable $V$ such that $X = g(Z) + V$ and $V \independent Z$.
In addition, an exact solution $\phidag$ and an error term $U$ 
depending on $V$ but not on $Z$ are chosen. Then $Y$ is defined as $Y := \phidag(X) + U$. 
With this construction both models \eqref{eq:iv_reg_mean_independence} and \eqref{eq:iv_reg_full_independence} identify the true solution. The functions and probability densities that were used for the test example are
\begin{align*}
 \phidag(x) &= \frac{1}{6}\sin(2\pi(x+0,25)),\\[6pt]
 f_Z(z) &= \frac{9}{7}\sqrt{z} + \frac{1}{7} \qquad \text{on the interval }[0,\, 1],\\[6pt]
 g(z) &= 0,8z+0,1,\\[6pt]
 f_V(v) &= \frac{1}{0,08\sqrt{2\pi}} \exp\left(-\frac{1}{2}\left(\frac{v}{0,08}\right)^2\right), \mbox{ and}\\[6pt]
 f_{U}(y,v) &= \frac{1}{0,07\sqrt{2\pi}} \exp\left(-\frac{1}{2}\left(\frac{y-2v}{0,07}\right)^2\right).
\end{align*}
The densities of $V$ and $U$ 
are constructed with Gaussians in a way that the expectation of $U$ 
depends on $v$. 
Figure \ref{fig:reg_without_W1d} shows the exact solution (blue) compared to the solution a nonparametric regression without instrumental variables would yield asymptotically (green).


\begin{center}
\includegraphics[width=0.5\textwidth]{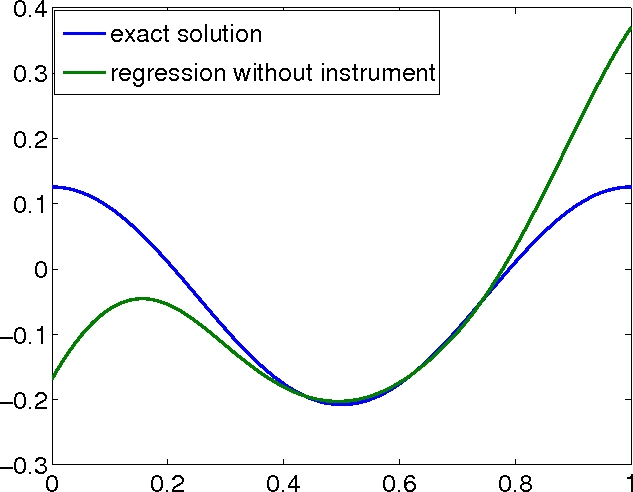}
\captionof{figure}{\label{fig:reg_without_W1d}Necessity of the instrument: A standard nonparametric regression would  asymptotically yield the green curve which is considerably different from the true curve $\varphi^\dag$ (blue).}
\end{center}


%
%

Both methods were tested on samples of $500$ and $1000$ observations. For each of the two sample sizes 1000 samples were generated and the joint density $f_{YXZ}$ was estimated by a kernel method. 

Figures \ref{fig:hist500} and \ref{fig:hist1000} show histograms for the $L^2$ error of the reconstructions for both methods and different sample sizes. The values are normed by the initial error, so that on this scale the initial error becomes $1$.
\begin{center}
\includegraphics[width=0.47\textwidth]{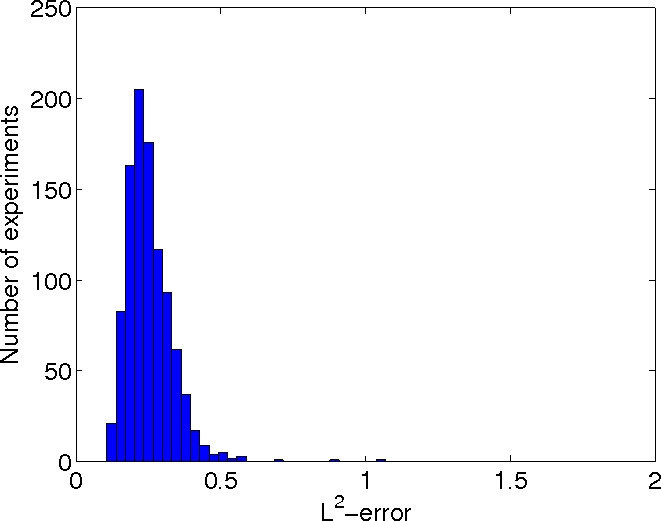}\qquad\includegraphics[width=0.47\textwidth]{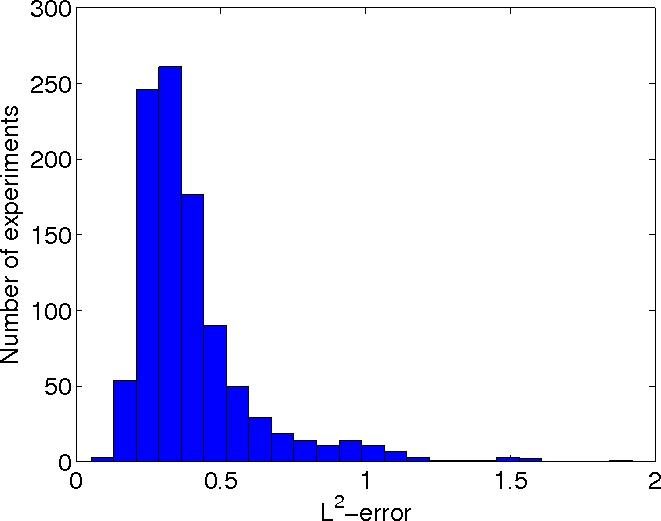}
\captionof{figure}{\label{fig:hist500}$L^2$ error for the sample size $n=500$. Left panel: IRGNM with the assumption $U \independent Z$. Right panel: iterated Tikhonov regularization with the assumption $\Ex[U|Z] = 0$}
\end{center}

In Figure \ref{fig:hist500} we compare the errors of both methods for the sample size $n=500$. Both methods produce acceptable results. The variance as well as the number of outliers observed for the method with independent instrument are significantly smaller than the variance or number of outliers of the method with the conditional mean assumption. The latter method produces a considerable number of outliers with the same or even larger errors than the initial guess. This cannot be observed for the IRGNM. In addition, the mean error of the IRGNM is smaller.
\begin{center}
\includegraphics[width=0.47\textwidth]{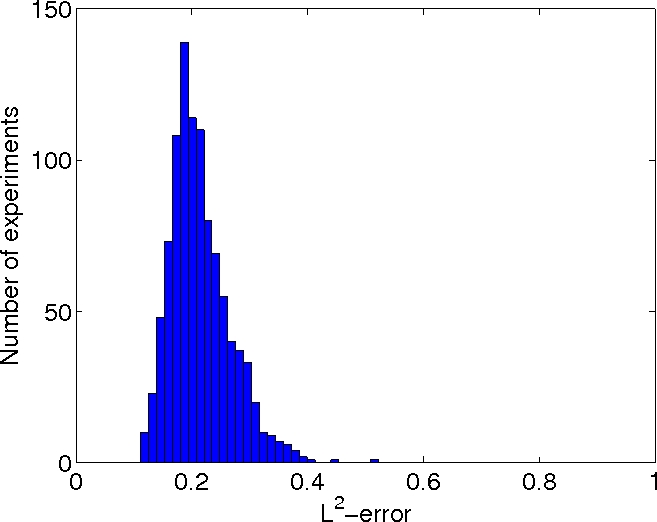}\qquad\includegraphics[width=0.47\textwidth]{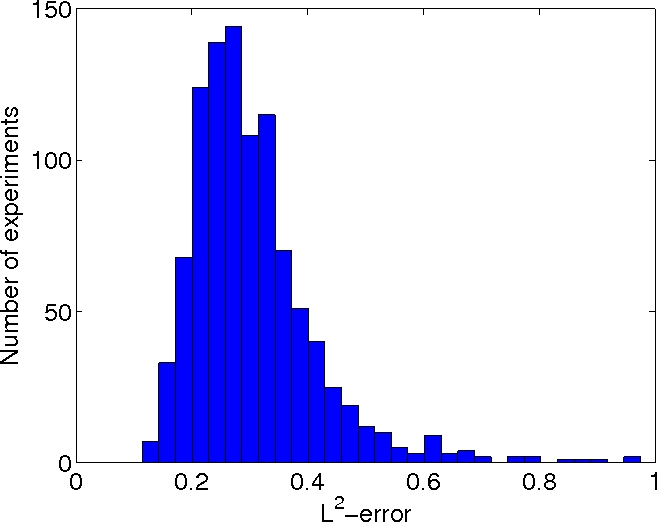}
\captionof{figure}{\label{fig:hist1000}$L^2$ error for the sample size $n=1000$. Left panel: IRGNM with the assumption $U \independent Z$. Right panel: iterated Tikhonov regularization with the assumption $\Ex[U|Z] = 0$}
\end{center}

Similar histograms for sample size $n=1000$ are displayed in Figures \ref{fig:hist1000}. Both methods perform well. The advantages of the IRGNM with less outliers, smaller variance and smaller mean error can be observed again. The following table provides the mean and some quantiles of the errors normed by the initial error.

\begin{small}
\begin{center}
\begin{tabular}{|l||l|ll|l|l|l|}
\hline
sample size and method & mean & quantiles\hspace{-7pt} & $q = 0.25$ & $q = 0.5$ & $q = 0.75$ & $q = 0.9$\\ \hline
$n = 500$, $\;\,U \independent W$ & 0.2535 & & 0.2012 & 0.2398 & 0.2940 & 0.3495\\ \hline
$n = 500$, $\;\,\Ex[U|W] = 0$ & 0.4042 & & 0.2738 & 0.3437 & 0.4475 & 0.6407\\ \hline
$n = 1000$, $U \independent W$ & 0.2152 & & 0.1780 & 0.2064 & 0.2439 & 0.2868\\ \hline
$n = 1000$, $\Ex[U|W] = 0$ & 0.3067 & & 0.2339 & 0.2846 & 0.3482 & 0.4325\\ \hline
\end{tabular}
\end{center}
\end{small}

We close this section with examples of median reconstructions for both sample sizes displayed in Figure \ref{fig:example}. They illustrate the advantage of the regression model with independent instruments solved with the IRGNM.

\begin{center}
\includegraphics[width=0.47\textwidth]{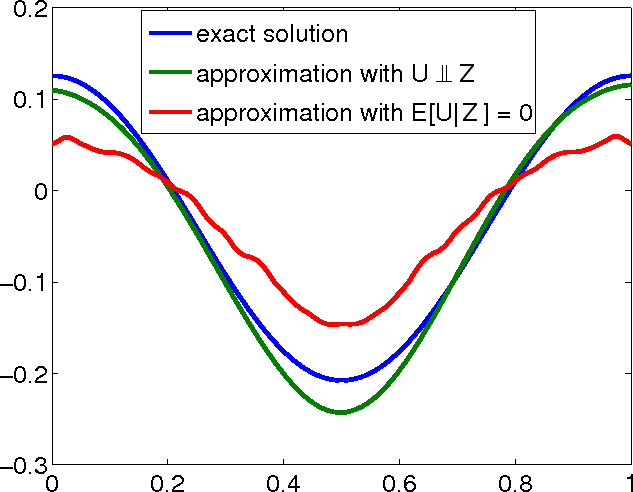}\qquad\includegraphics[width=0.47\textwidth]{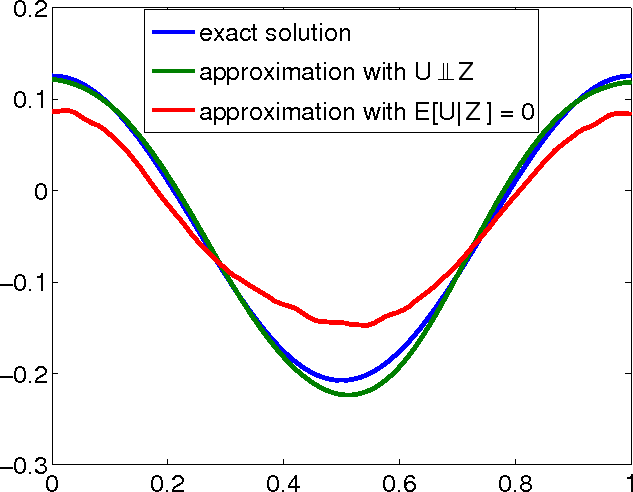}
\captionof{figure}{\label{fig:example}Examples for reconstructions with sample size $n=500$ (left) and $n=1000$ (right). The blue line shows the exact solution, the red curve the reconstruction with the conditional mean assumption and the green curve the reconstruction with independent instrument.}
\end{center}

These results suggest that both methods give consistent estimators for the nonparametric instrumental regression with clear advantages for the regression model with independent instruments \eqref{eq:iv_reg_full_independence} and the IRGNM.

\section*{Acknowledgment}
The author would like to thank Thorsten Hohage and Johannes Schmidt-Hieber for interesting and fruitful discussions on this topic. He also would like to thank two anonymous referees for valuable comments that improved the paper. 

\bibliography{lit}
\bibliographystyle{apalike}

\appendix
\section{Appendix}\label{sec:appenix}

\subsection{Nonlinearity restriction}
We prove in this section that for the IV regression examples Assumption \ref{ass:lipschitz_alt} is an alternative to 
Assumption \ref{ass:lipschitz}.
\begin{lemma}\label{lem:nonlin_restrict}
For the operator equations \eqref{eq:nonlininstreg:opeq2} and \eqref{eq:iv_quant_reg_op}  Assumption \ref{ass:lipschitz_alt} implies Assumption \ref{ass:lipschitz}.
\end{lemma}

\begin{proof}
Let $\widehat k_n(\xi_1(x),x,z)$ denote the estimate for $k_{ind}$ or $k_q$. Since the second derivatives with respect to $y$ are bounded we have
\begin{equation*}
\begin{split}
\|\widehat \Op_n'[\xi_1] - \Ophat_n'[\xi_2]\|_{\calL(\X,\Y)} &\leq 
\sqrt{\iint \left( \frac{\partial}{\partial y} \widehat k_n(\xi_1(x),x,z) - \frac{\partial}{\partial y} \widehat k_n(\xi_2(x),x,z) \right)^2 dx\, dz}\\
&\leq \sqrt{\iint \left( \sup_{y,\widetilde x} \left|\frac{\partial^2}{\partial y^2} \widehat k_n(y,\widetilde x, z)\right| \left(\xi_1(x) - \xi_2(x)\right) \right)^2 dx\, dz}\\
&=  \mu(\supp(Z)) \sup_{y, x, z} \left|\frac{\partial^2}{\partial y^2} \widehat k_n(y, x, z) \right| \, \|\xi_1 - \xi_2\|_\X
\end{split}
\end{equation*}
If $\widehat k_n$ is a strongly consistent estimator, for any constant $c > 0$
\[                                                                                                                                 
\sup_{y, x, z} \left| \frac{\partial^2}{\partial y^2} \widehat k_n(y, x, z) \right| \leq \sup_{y, x, z} \left| \frac{\partial^2}{\partial y^2} k(y, x, z) \right| + c\qquad \mbox{almost surely for large n.}
\]
Hence, Assumption \ref{ass:lipschitz} holds with
\[
L = \mu(\supp(Z))\sup_{y, x, z} \left| \frac{\partial^2}{\partial y^2} k(y, x, z) \right| + c
\]
for any $c>0$.
\end{proof}

\subsection{Concentration inequalities}
This section proves Lemma \ref{lem:concentration} using McDiamid's extension of Hoeffding's inequality.
\begin{theorem}[\cite{McD:89}]\label{the:McD}
Let $W_1, \ldots, W_n$ be independent random variables. If $f : \supp(W_1, \ldots, W_n) \rightarrow \R$ satisfies for $1 \le i \le n$
\begin{align}\label{eq:mcdiamid}
\sup_{\substack{(w_1,\ldots,w_n), (w_1',\ldots,w_n')\\ \in \supp(W_1, \ldots, W_n)}} \left|f(w_1,\ldots,w_n) - f(w_1,\ldots,w_{i-1},w_i',w_{i+1},\ldots,w_n) \right| \le c_i.
\end{align}
Then for $\tau \ge 0$
\[
\Prob\{|f(W_1,\ldots,W_n) - \Ex f(W_1,\ldots,W_n)| \ge \sqrt{\tau}\} \le 2\exp\left(\frac{-2\tau}{\sum_{i=1}^n c_i^2}\right).
\]
\end{theorem}

Now we can prove Lemma \ref{lem:concentration}.

\begin{proof}(of Lemma \ref{lem:concentration})
Denote the kernels by
\[
K_{Y,h}(y) = \frac{1}{h} K_Y\left(\frac{y}{h}\right),\qquad K_{X,h}(x) = \frac{1}{h^{d_X}} K_X\left(\frac{x}{h}\right),\qquad K_{Z,h}(z) = \frac{1}{h^{d_Z}} K_Z\left(\frac{z}{h}\right).
\]
The proof has 4 parts in which we prove for each of the operators  \eqref{eq:nonlininstreg:opeq2} and \eqref{eq:iv_quant_reg_op} the inequalities \eqref{eq:concentration_op} and \eqref{eq:concentration_de}.
\paragraph{part 1} \eqref{eq:concentration_op} for operator \eqref{eq:nonlininstreg:opeq2}

We show \eqref{eq:mcdiamid} with $W = (Y,X,Z)$ and 
\begin{align*}
&f\big((y_1,x_1,z_1),\ldots, (y_n,x_n,z_n)\big) := \|\Ophat_{ind}(\phidag)(u,z)\|_{L_2(u,z)} \\
&=\Big\|\int n^{-1}\sum_{k=1}^n K_{Y,h}(\phidag(x)-u-y_k)K_{X,h}(x-x_k) K_{Z,h}(z - z_k)\\
& \qquad\quad - n^{-2}\sum_{k=1}^n \sum_{l=1}^n K_{Y,h}(\phidag(x)-u-y_k)K_{X,h}(x-x_k) K_{Z,h}(z - z_l) dx\Big\|_{L_2(u,z)}.
\end{align*}
Iterated application of the triangular inequality and dropping the terms that cancel in the sums yields
\begin{align*}
&|f(w_1,\ldots,w_n) - f(w_1,\ldots,w_{i-1},w_i',w_{i+1},\ldots,w_n)|\\
&\le n^{-1} \Big\|\int K_{Y,h}(\phidag(x)-u-y_i)K_{X,h}(x-x_i) K_{Z,h}(z - z_i)\\
& \qquad \qquad - K_{Y,h}(\phidag(x)-u-y'_i)K_{X,h}(x-x'_i) K_{Z,h}(z - z'_i)dx\Big\|_{L_2(u,z)}\\
&+ n^{-2} \Big\|\int \sum_{k\neq i}K_{Y,h}(\phidag(x)-u-y_k)K_{X,h}(x-x_k)\big[ K_{Z,h}(z - z'_i)-K_{Z,h}(z - z_i)\big]dx\Big\|_{L_2(u,z)}\\
&+ n^{-2} \Big\|\int \sum_{l\neq i}\big[K_{Y,h}(\phidag(x)-u-y'_i)K_{X,h}(x-x'_i)\\
&\qquad \qquad\qquad\quad - K_{Y,h}(\phidag(x)-u-y_i)K_{X,h}(x-x_i)\big]K_{Z,h}(z - z_l)dx\Big\|_{L_2(u,z)}\\
&+ n^{-2} \Big\|\int K_{Y,h}(\phidag(x)-u-y'_i)K_{X,h}(x-x'_i) K_{Z,h}(z - z'_i)\\
&\qquad\qquad-K_{Y,h}(\phidag(x)-u-y_i)K_{X,h}(x-x_i)K_{Z,h}(z - z_i)dx\Big\|_{L_2(u,z)}
\end{align*}
We use the triangular inequality again and substitute $x_i,y_i,z_i$ for $x'_i,y'_i,z'_i$ and for $x_k,y_k,z_l$.
\begin{align*}
&\le 2n^{-1} \Big\|\int K_{Y,h}(\phidag(x)-u-y_i)K_{X,h}(x-x_i) K_{Z,h}(z - z_i)dx\Big\|_{L_2(u,z)}\\
&+ 2n^{-2} \Big\|\int \sum_{k\neq i}K_{Y,h}(\phidag(x)-u-y_k)K_{X,h}(x-x_k)K_{Z,h}(z - z_i)dx\Big\|_{L_2(u,z)}\\
&+ 2n^{-2} \Big\|\int \sum_{l\neq i} K_{Y,h}(\phidag(x)-u-y_i)K_{X,h}(x-x_i)K_{Z,h}(z - z_l)dx\Big\|_{L_2(u,z)}\\
&+ 2n^{-2} \Big\|\int K_{Y,h}(\phidag(x)-u-y'_i)K_{X,h}(x-x'_i) K_{Z,h}(z - z'_i)dx\Big\|_{L_2(u,z)}\\
&\le \left(\frac{2}{n}+ \frac{4(n-1)}{n^2} + \frac{2}{n^2}\right)\Big\|\int K_{Y,h}(\phidag(x)-u-y'_i)K_{X,h}(x-x'_i) K_{Z,h}(z - z'_i)dx\Big\|_{L_2(u,z)}\\
&< 6n^{-1}\Big\|\int K_{Y,h}(\phidag(x)-u-y'_i)K_{X,h}(x-x'_i) K_{Z,h}(z - z'_i)dx\Big\|_{L_2(u,z)}\\
&= 6n^{-1}  \bigg(\int\Big(\int K_{Y,h}(\phidag(hx+x_i)-u-y_i)K_{X}(x) K_{Z,h}(z - z_i)dx\Big)^2 d(u,z)\bigg)^{1/2}\\
&\le 6n^{-1}  \bigg(\int K_{X}^2(x)\int K_{Y,h}^2(\phidag(hx+x_i)-u-y_i) K_{Z,h}^2(z - z_i) d(u,z)\,dx\bigg)^{1/2}\\
&= 6n^{-1}  \bigg(\int K_{X}^2(x) \|K_{Y,h}(u) K_{Z,h}(z)\|_{L^2(u,z)}^2dx\bigg)^{1/2}\\
&= 6n^{-1}  \|K_{X}\|_{L^2} \|K_{Y,h}(u) K_{Z,h}(z)\|_{L^2(u,z)}\\
&= 6 n^{-1}h^{-(d_Z+1)/2}\|K_{X}\|_{L^2} \|K_{Y}(u)K_{Z}(z)\|_{L^2(u,z)}.
\end{align*}
We used the standard substitution arguments for kernel methods to get from $K_{X,h}$ to $K_{X}$ and from $K_{Y,h}K_{Z,h}$ to $K_{Y}K_{Z}$. Together with Theorem \ref{the:McD} this proves
\begin{align}\label{eq:var_Find}
\Prob \left\{ \left| \|\Ophat_{ind}(\phidag)\|_{L_2} - \Ex\|\Ophat_{ind}(\phidag)\|_{L_2} \right| \ge \sqrt\tau \right\} \le 2\exp\left(\frac{-\tau n h^{d_Z+1}}{18\|K_{X}\|_{L^2}^2\|K_{Y}K_{Z}\|_{L_2}^2}\right).
\end{align}
Hence, there exists a constant $c_2 > 0$ such that
\[
\Prob \left\{ \left| \|\Ophat_{ind}(\phidag)\|_{L_2} - \Ex\|\Ophat_{ind}(\phidag)\|_{L_2} \right| \ge \sqrt{\tau \Var\left(\|\Ophat_{ind}(\phidag )\|_{L_2} \right)}\right\} \le 2\exp\left(- c_2\tau \right).
\]


\paragraph{part 2}  \eqref{eq:concentration_op} for operator \eqref{eq:iv_quant_reg_op}

A similar argument applies to the quantile regression operator in \eqref{eq:iv_quant_reg_op}. We set $\bar{K}_{Y,h}(y) = \int_{-\infty}^{y} K_{Y}(\tilde y) d\tilde y$ and $\bar C_Y := |\sup_y \bar K_{Y,h}(y)-\inf_t \bar K_{Y,h}(y)|= |\sup_y \bar K_{Y,1}(y)-\inf_y \bar K_{Y,1}(y)|$.
Note that $\bar C_Y$ does not depend on $h$. Theorem \ref{the:McD} is now applied with
\begin{align*}
f&(w_1,\ldots,w_n) = \|\Ophat_q(\phidag)\|_{L_2}\\
& = \Big\|n^{-1}\int \bar K_{Y,h}(\phidag(x)-y_i)K_{X,h}(x-x_i) K_{Z,h}(z - z_i)dx -q K_{Z,h}(z - z_i)\Big\|_{L_2(z)}.
\end{align*}
The estimation
\begin{align*}
&|f(w_1,\ldots,w_n) - f(w_1,\ldots,w_{i-1},w_i',w_{i+1},\ldots,w_n)|\\
&= n^{-1} \Big\|\int \bar K_{Y,h}(\phidag(x)-y_i)K_{X,h}(x-x_i) K_{Z,h}(z - z_i)dx -q K_{Z,h}(z - z_i) \\
& \qquad \qquad - \int \bar K_{Y,h}(\phidag(x)-y'_i)K_{X,h}(x-x'_i) K_{Z,h}(z - z'_i)dx + q K_{Z,h}(z - z_i')\Big\|_{L_2}\\
&\le n^{-1} \Big\|\bar C_Y\int K_{X,h}(x-x_i) K_{Z,h}(z - z_i)dx\Big\|_{L_2} + 2qn^{-1} \Big\|K_{Z,h}(z - z_i)\Big\|_{L_2}\\
& = \bar C_Y (1+2q) n^{-1} \|K_{Z,h}\|_{L_2}\\
& = \bar C_Y (1+2q) \|K_{Z}\|_{L_2} n^{-1}h^{-d_Z/2}.
\end{align*}
proves together with Theorem \ref{the:McD}
\begin{align}\label{eq:var_Fq}
\Prob \left\{ \left| \|\Ophat_q(\phidag)\|_{L_2} - \Ex\|\Ophat_q(\phidag)\|_{L_2} \right| \ge \sqrt\tau \right\} \le 2\exp\left(\frac{-2\tau n h^{d_Z}}{\bar C_Y^2 (1+2q)^2 \|K_{Z}\|_{L_2}^2}\right).
\end{align}
Hence, there exists a constant $c_2 > 0$ such that \eqref{eq:concentration_de} holds for $\Ophat_q$.
%
%
%
\paragraph{part 3} \eqref{eq:concentration_de} for operator \eqref{eq:nonlininstreg:opeq2}

We follow the same strategy and adopt the notation above. Theorem \ref{the:McD} is applied with $W = (Y,X,Z)$ 
with
\begin{align*}
&f\big((y_1,x_1,z_1),\ldots, (y_n,x_n,z_n)\big) := \|\Tdaghat - \Tdag\|_{HS}^{1+\mu} \\
&=\bigg\|n^{-1}\sum_{i=1}^n K_{Y,h}'(\phidag(x)-u-y_i)K_{X,h}(x-x_i) K_{Z,h}(z - z_i) \\
& \qquad\quad - n^{-2}\sum_{k=1}^n \sum_{l=1}^n K_{Y,h}'(\phidag(x)-u-y_k)K_{X,h}(x-x_k) K_{Z,h}(z - z_l)\\
&\hspace{70pt}- \frac{\partial}{\partial y} f_{YXZ}(\phidag(x)-u,x,z) + \frac{\partial}{\partial y} f_{YX}(\phidag(x)-u,x)f_Z(z) \bigg\|_{L_2}^{1+\mu}
\end{align*}
where $K_{Y,h}'$ is the derivative of $K_{Y,h}$. With the same steps as in part 1 with repeated application of the triangular inequality and substitution we get
\begin{align*}
&|f(w_1,\ldots,w_n) - f(w_1,\ldots,w_{i-1},w_i',w_{i+1},\ldots,w_n)|\\
&\le   \big\| n^{-1}K'_{Y,h}(\phidag(x)-u-y_i)K_{X,h}(x-x_i) K_{Z,h}(z - z_i)\\
&\quad \quad - n^{-1}K'_{Y,h}(\phidag(x)-u-y'_i)K_{X,h}(x-x'_i) K_{Z,h}(z - z'_i)\\
&\quad\quad - n^{-2}\sum_{k=1}^n \sum_{l=1}^n K_{Y,h}'(\phidag(x)-u-y_k)K_{X,h}(x-x_k) K_{Z,h}(z - z_l)\\
&\quad\quad - n^{-2}\sum_{k=1}^n \sum_{l=1}^n K_{Y,h}'(\phidag(x)-u-y_k')K_{X,h}(x-x_k') K_{Z,h}(z - z_l')\big\|_{L_2}^{1+\mu}\\
&\le 6^{1+\mu} n^{-1-\mu} \big\| K'_{Y,h}(\phidag(x)-u-y_i)K_{X,h}(x-x_i) K_{Z,h}(z - z_i)\big\|_{L_2}^{1+\mu}\\
&= 6^{1+\mu} n^{-1-\mu} \| K'_{Y,h}\|_{L_2}^{1+\mu}\, \|K_{X,h}\|_{L_2}^{1+\mu}\, \|K_{Z,h}\|_{L_2}^{1+\mu}\\
&= 6^{1+\mu} n^{-1-\mu}h^{-\frac{(1+\mu)(d_X+d_Z+3)}{2}} \| K'_{Y}\|_{L_2}^{1+\mu}\, \|K_{X}\|_{L_2}^{1+\mu}\, \|K_{Z}\|_{L_2}^{1+\mu}.
\end{align*}
Hence,
\begin{align}\label{eq:iv_reg_der_tail}
\Prob \left\{ \left| \|\Tdaghat - \Tdag\|_{HS}^{1+\mu} - \Ex\left(\|\Tdaghat - \Tdag\|_{HS}^{1+\mu}\right) \right| \ge \sqrt\tau \right\}\qquad\qquad\\
\le 2\exp\left(\frac{-2\tau n^{1+2\mu} h^{(1+\mu)(d_X+d_Z+3)}}{36^{1+\mu}\|K_{X}\|_{L^2}^{2+2\mu}\|K_{Y}\|_{L^2}^{2+2\mu}\|K_{Z}\|_{L_2}^{2+2\mu}}\right).
\end{align}
Thus, there exist a constant $c_4$ such that
\begin{align*}
\Prob \left\{ \left| \|\Tdaghat - \Tdag\|_{HS}^{1+\mu} - \Ex\left(\|\Tdaghat - \Tdag\|_{HS}^{1+\mu}\right) \right|\ge \sqrt{\tau \Var\left(\|\Tdaghat - \Tdag\|_{HS}^{1+\mu}\right)} \right\}\qquad\\
\le 2\exp\left(-c_4 \tau\right).
\end{align*}

\paragraph{part 4} \eqref{eq:concentration_de} for operator \eqref{eq:iv_quant_reg_op}

A similar argument holds for the instrumental quantile regression problem. The kernel of the Fr\'echet derivative of the operator $\Op_q$ in \eqref{eq:iv_quant_reg_op} at $\phidag$ is simply $f_{YXZ}(\phidag(x),x,z)$. So Theorem \ref{the:McD} is applied to
 \begin{align*}
&f\big((y_1,x_1,z_1),\ldots, (y_n,x_n,z_n)\big) := \|\Tdaghat - \Tdag\|_{HS}^{1+\mu} \\
&=\left\|n^{-1}\sum_{i=1}^n K_{Y,h}(\phidag(x)-y_i)K_{X,h}(x-x_i) K_{Z,h}(z - z_i) - f_{YXZ}(\phidag(x),x,z) \right\|_{L_2}^{1+\mu}.
\end{align*}
Note that $\sup_y\big(K_{Y,h}(y)\big) - \inf_y \big(K_{Y,h}(y)\big) = h^{-1}\left[\sup_y\big(K_{Y}(y)\big) - \inf_y \big(K_{Y}(y)\big)\right]$ and set $C_Y = \left|\sup_y\big(K_{Y}(y)\big) - \inf_y \big(K_{Y}(y)\big)\right|$. This allows for the following estimation
\begin{align*}
&|f(w_1,\ldots,w_n) - f(w_1,\ldots,w_{i-1},w_i',w_{i+1},\ldots,w_n)|\\
&\le  n^{-1-\mu} \big\| K_{Y,h}(\phidag(x)-y_i)K_{X,h}(x-x_i) K_{Z,h}(z - z_i)\\
& \qquad \qquad - K_{Y,h}(\phidag(x)-y'_i)K_{X,h}(x-x'_i) K_{Z,h}(z - z'_i)\big\|_{L_2}^{1+\mu}\\
& \le n^{-1-\mu} h^{-1-\mu} C_Y \big\|K_{X,h}(x-x'_i) K_{Z,h}(z - z'_i)\big\|_{L_2}^{1+\mu}\\
& = n^{-1-\mu} h^{-\frac{(1+\mu)(2+d_X+d_Z)}{2}} C_Y \|K_{X}\|_{L_2}^{1+\mu} \|K_{Z}\|_{L_2}^{1+\mu}.
\end{align*}
This implies
\begin{align}\label{eq:iv_quant_der_tail}
\Prob \left\{ \left|  \|\Tdaghat - \Tdag\|_{HS}^{1+\mu} - \Ex\left(\|\Tdaghat - \Tdag\|_{HS}^{1+\mu}\right)  \right| \ge \sqrt\tau \right\}\qquad\qquad\\
\le 2\exp\left(\frac{-2\tau n^{1+2\mu} h^{(1+\mu)(2+d_X+d_Z)}}{C_Y^2\|K_{X}\|_{L_2}^{2+2\mu} \|K_{Z}\|_{L_2}^{2+2\mu} }\right).
\end{align}
Hence, there exist a constants $c_4$ such that 
\begin{align*}
\Prob \left\{ \left| \|\Tdaghat - \Tdag\|_{HS}^{1+\mu} - \Ex\left(\|\Tdaghat - \Tdag\|_{HS}^{1+\mu}\right) \right| \ge \sqrt{\tau \Var\left(\|\Tdaghat - \Tdag\|_{HS}^{1+\mu}\right)} \right\}\qquad\\
\le 2\exp\left(-c_4 \tau\right).
\end{align*}
\end{proof}

\begin{corollary}\label{cor:var_noi}
Under the assumptions of Lemma \ref{lem:concentration}

For the operator $\Op_{ind}$ in \eqref{eq:nonlininstreg:opeq2}
\begin{align*}
\Var\left(\|\Ophat_{ind}(\phidag )\|_{L_2} \right) &\le \frac{18\|K_{X}\|_{L^2}^2\|K_{Y}K_{Z}\|_{L_2}^2}{n h^{d_Z+1}} = \Or(n^{-1}h^{-d_Z-1})\\
\Var\left( \|\Tdaghat - \Tdag\|_{HS}^{1+\mu} \right) &\le \frac{36^{1+\mu}\|K_{X}\|_{L^2}^{2+2\mu}\|K_{Y}\|_{L^2}^{2+2\mu}\|K_{Z}\|_{L_2}^{2+2\mu}}{2n^{1+2\mu} h^{(1+\mu)(d_X+d_Z+3)}}.
\end{align*}

For the operator $\Op_{q}$ in \eqref{eq:iv_quant_reg_op}
\begin{align*}
\Var\left(\|\Ophat_{q}(\phidag )\|_{L_2} \right) &\le \frac{\bar C_Y^2 (1+2q)^2 \|K_{Z}\|_{L_2}^2}{2 n h^{d_Z}}  = \Or(n^{-1}h^{-d_Z})\\
\Var\left( \|\Tdaghat - \Tdag\|_{HS}^{1+\mu} \right) &\le \frac{C_Y^2\|K_{X}\|_{L_2}^{2+2\mu} \|K_{Z}\|_{L_2}^{2+2\mu} }{2 n^{1+2\mu} h^{(1+\mu)(2+d_X+d_Z)}}
\end{align*}
\end{corollary}

\begin{proof}
It follows from \eqref{eq:var_Find} that $\|\Ophat_{ind}(\phidag )\|_{L_2}$ is a subgaussian random variable. The moment condition for subgaussian variables implies that
\[
\Var\left(\|\Ophat_{ind}(\phidag )\|_{L_2} \right) = \frac{18\|K_{X}\|_{L^2}^2\|K_{Y}K_{Z}\|_{L_2}^2}{n h^{d_Z+1}}~.
\]
The other three statements follow in the same way from \eqref{eq:iv_reg_der_tail}, \eqref{eq:var_Fq}, and \eqref{eq:iv_quant_der_tail} respectively.
\end{proof}

\subsection{Error analysis}\label{sec:error}

In this section we prepare the proofs of the convergence rate theorems by decompose the error $e_{j+1} := \phihat_{j+1} - \phidag$ of our method \eqref{eq:irgnm_spektral} into different components and derive estimates for each component.

\subsubsection{Error decomposition}
The error in the $j+1$-th Newton step is
\begin{equation*}
\begin{split}
e_{j+1} &= \phihat_{j+1} - \phidag= \varphi_0 - \phidag + g_{\alpha_j}(\That_{n,j}^* \That_{n,j})\That_{n,j}^* \left(\That_{n,j}(\phihat_j-\varphi_0) - \Ophat_n(\phihat_j) \right).
\end{split}
\end{equation*}
We decompose the error into four parts. These are an approximation error, a propagated noise error, an error due to noise in the derivative, and a nonlinearity error
\[
e_{j+1} = e^{app}_{j+1} + e^{noi}_{j+1} + e^{der}_{j+1} + e^{nl}_{j+1}.
\]
In the decomposition we use a function $r_\alpha$ defined as $r_\alpha(\lambda) := 1 - \lambda g_\alpha(\lambda)$. With $g_\alpha$ as in \eqref{eq:g_alpha} we have: $r_\alpha(\lambda) = \left(\frac{\alpha}{\lambda + \alpha} \right)^m$.
\begin{description}
 \item[approximation error] $e^{app}_{j+1} := r_{\alpha_j}(\Tdaghat^* \Tdaghat) \Lambda(\Tdaghat^* \Tdaghat) \omega$

 \item[propagated noise error] $e^{noi}_{j+1} := g_{\alpha_j}(\That_{n,j}^*\That_{n,j})\That_{n,j}^*[
 - \Ophat_n(\phidag)]$

 \item[derivative noise error] $e^{der}_{j+1} := r_{\alpha_j}(\That_{n,j}^*\That_{n,j})[\Lambda(\Tdag^* \Tdag) - \Lambda(\Tdaghat^* \Tdaghat)] \omega$

 \item[nonlinearity error] $e^{nl}_{j+1} := g_{\alpha_j}(\That_{n,j}^*\That_{n,j})\That_{n,j}^*[\Ophat_n(\phidag) - \Ophat_n(\phihat_j) + \That_{n,j}(\phihat_j - \phidag)]$ 
 
\hspace{170pt} $+ [r_{\alpha_j}(\That_{n,j}^*\That_{n,j}) - r_{\alpha_j}(\Tdaghat^* \Tdaghat)]\Lambda(\Tdaghat^* \Tdaghat) \omega$
\end{description}
A similar related decomposition without $e^{der}_{j+1}$ was proposed in \cite{Baku:92} for the case of known operators. In the rest of the section we will derive bounds on each error component.

\subsubsection{Approximation error}

Assumption \ref{ass:saturation} implies the existence of a constant $C_\Lambda$ such that
\[
\sup_{0<x\le \|\Tdaghat^* \Tdaghat\|
} r_\alpha(x) \Lambda(x) \leq C_\Lambda \Lambda(\alpha) \qquad \text{ for all } \alpha \geq 0.
\]
Hence, with $\rho = \|\omega\|$ the approximation error is bounded by
\begin{equation}\label{eq:eapp}
\|e^{app}_{j+1}\|_\X \leq C_\Lambda \Lambda(\alpha_j) \rho.
\end{equation}
Furthermore, in our setting with $\alpha_j := q_\alpha \alpha_{j-1}$ the following inequalities hold with $\gamma_{app} := q_\alpha^{-m}$
\begin{equation}\label{eq:gammaapp}
\begin{alignedat}{2}
&\|e^{app}_{j+1}\|_\X \leq \|e^{app}_j\|_\X \leq \gamma_{app} \|e^{app}_{j+1}\|_\X &\qquad &\text{for } j \geq 1\\
&\text{and } \|e^{app}_0\|_\X \leq \gamma_{app}\|e^{app}_1\|_\X &\qquad &\text{since } \alpha_0 \geq \frac{\|\Tdaghat^* \Tdaghat\|_{\calL(\X,\X)}}{1-q_\alpha}.
\end{alignedat}
\end{equation}
Note that the bound on the approximation error behave like a bias term. It tends to $0$ with increasing $j$ because $\alpha_j$ is decreasing while $\Lambda$ is strictly increasing and $\Lambda(0) = 0$.

\subsubsection{Propagated noise error}

The propagated noise error $e^{noi}_{j+1} := g_{\alpha_j}(\That_{n,j}^*\That_{n,j})\That_{n,j}^* [-\Ophat_n(\phidag)]$ can be bounded by using some standard estimates and the functional calculus. Note that for any linear bounded operator $T: \X \rightarrow \Y$ and $\psi \in \Y$
\begin{align}
\nonumber\|g_\alpha(TT^*)\|_{\mathcal{L}(\Y,\Y)} &\leq \|g_\alpha\|_\infty = \sup_{x \geq 0} \left( \frac{(x+\alpha)^m -\alpha^m}{x(x+\alpha)^m}\right) \leq \frac{m}{\alpha}~,\\
\nonumber\|g_\alpha(TT^*)TT^*\|_{\mathcal{L}(\Y,\Y)} &\leq \sup_{x \geq 0} |g_\alpha(x)x| = \sup_{x \geq 0} \left( \frac{(x+\alpha)^m -\alpha^m}{(x+\alpha)^m}\right) = 1~, \text{ and}\\
\begin{split}\label{eq:g(TT)T}
\|g_\alpha(T^*T)T^*\psi\|_\X^2 &= \langle g_\alpha(TT^*)\psi,\, g_\alpha(TT^*)TT^*\psi \rangle_\Y\\
&\leq \sup_{x \geq 0} |x g_\alpha(x)| \|g_\alpha\|_\infty \|\psi\|_\Y^2 \leq \frac{m}{\alpha} \|\psi\|_\Y^2.
\end{split}
\end{align}
where $\|\cdot\|_\infty$ denotes the sup norm. Hence, 
\begin{align}
\nonumber&\|e^{noi}_{j+1}\|_\X = \|g_{\alpha_j}(\That_{n,j}^*\That_{n,j})\That_{n,j}^* \Ophat_n (\phidag)\|_\X \leq \sqrt{\frac{m}{\alpha_j}} \|\Ophat_n(\phidag)\|_\Y \quad \mbox{and}\\
\label{eq:est_noi_1}&\Ex \left(\|e^{noi}_{j+1}\|_\X^2\right) \le \frac{m}{\alpha_j} \Ex\left(\|\Ophat_n(\phidag)\|^2_\Y\right).
\end{align}
Note that this bound does not depend on the noise in the derivative $\That_{n,j}$ but only on the error in the operator $\Ophat_n$ at $\phidag$. It behaves like a variance term, i.e. it increases with a decreasing regularization parameter $\alpha_j$. In addition to the bound \eqref{eq:est_noi_1} a concentration inequality is needed to bound the MISE of the Newton method. This is part of Assumption \ref{ass:concentration}. The following example illustrates the asymptotic behavior of the bound \eqref{eq:est_noi_1} for the nonparametric IV operators \eqref{eq:nonlininstreg:opeq2} and \eqref{eq:iv_quant_reg_op}.

\begin{example}\label{example_noi}
We work under the assumptions of Lemma \ref{lem:concentration}. We need different smoothness condition for the operators $\Op_{ind}$ and $\Op_q$. Assume for $\Op_{ind}$ that all derivatives of degree $r$ of the density $f_{YXZ}$ exist and are bounded. For the operator in \eqref{eq:iv_quant_reg_op} we assume less smoothness. Derivatives of degree $r$ of $F_{YXZ}$ should exist and be bounded. Let the joint density $f_{YXZ}$ be estimated by a kernel density estimator $\widehat f_{YXZ}$ with a kernel of sufficiently high order and with a common bandwidth $h$. This gives straight forward estimators $\widehat k_{ind}$ and $\Ophat_{ind}$ for the regression with full independence and $\widehat k_q$ and $\Ophat_q$ for the quantile regression.

With these smoothness assumptions, sample size $n$, and bandwidth $h$ the estimators $\widehat k_{ind}$ and $\widehat k_q$ converge in both cases with the rate 
\[
\Ex(\|k - \widehat k\|_{L^2}^2) = \Or(n^{-1} h^{-d_X-d_Z-1} + h^{2r}).
\]
It follows from Corollary \ref{cor:var_noi} that
\begin{align*}
&\Ex(\|\Ophat_{ind}(\varphi) - \Op_{ind}(\varphi)\|_{L^2}^2) = \Or(n^{-1} h^{-d_Z-1} + h^{2r})\\
&\phantom{\Ex(\|\Ophat_{ind}(\varphi) - \Op_{ind}(\varphi)\|_{L^2}^2)} = \Or(n^{-\frac{2r}{2r+d_Z+1}}) \quad \text{when } h \sim n^{-\frac{1}{2r+d_Z+1}} \qquad \mbox{and}\\
&\Ex(\|\Ophat_q(\varphi) - \Op_q(\varphi)\|_{L^2}^2) = \Or(n^{-1} h^{-d_Z} + h^{2r}) = \Or(n^{-\frac{2r}{2r+d_Z}}) \quad \text{when } h \sim n^{-\frac{1}{2r+d_Z}}.
\end{align*}
With the bound in \eqref{eq:est_noi_1} the rate of the MISE of the propagated noise error is
\begin{align*}
\Ex\left(\|e^{noi}_{j+1}\|_{L^2}^2\right) &\leq \frac{m}{\alpha_j} \Ex\left(\|\Ophat_{ind}(\phidag)\|_{L^2}^2\right)= \Or\left(\alpha_j^{-1} (n^{-1} h^{-d_Z-1} + h^{2r})\right) \text{ for } \Op_{ind},\\
\Ex\left(\|e^{noi}_{j+1}\|_{L^2}^2\right) &\leq \frac{m}{\alpha_j} \Ex\left(\|\Ophat_q(\phidag)\|_{L^2}^2\right)= \Or\left(\alpha_j^{-1} (n^{-1} h^{-d_Z} + h^{2r})\right) \text{ for } \Op_q.
\end{align*}
\end{example}

\subsubsection{Derivative noise error}

The simple observation that $r_\alpha(x) = \left(\frac{\alpha}{x + \alpha}\right)^m \leq 1$ for $x \in [0,\, \infty)$ independent of $\alpha$ or $m$ leads to the estimate $\|e^{der}_{j+1}\|_\X \leq \rho \|\Lambda(\Tdag^* \Tdag) - \Lambda(\Tdaghat^* \Tdaghat)\|_{\calL(\X,\X)}$.
Here the norm on the right hand side of the inequality is the usual operator norm. A way to simplify the term $\|\Lambda(\Tdag^* \Tdag) - \Lambda(\Tdaghat^* \Tdaghat)\|_{\calL(\X,\X)}$ is provided by the following lemma.

\begin{lemma}[\cite{Egger:05} Lemma 3.2.]\label{lemma:egger}
For two linear bounded operators between Hilbert spaces $A$ and $B$ and $\mu > \frac{1}{2}$ there exists a constant $c_\mu$ such that with the corresponding operator norms
\begin{align*}
\|(A^*A)^\mu - (B^*B)^\mu\| \leq c_\mu \|A - B\|\, \big|\|A\| - \|B\|\big|^\mu~.
\end{align*}
\end{lemma}
Hence, with some constant $C_d$ and a norm $\|\cdot\|_D$ which is either the operator norm or some norm dominating the operator norm
\begin{align}\label{eq:der_noi_est}
\|e^{der}_{j+1}\|_\X \leq C_d \rho \|\Tdaghat - \Tdag\|_D^{1 + \mu}~.
\end{align}
The bound in \eqref{eq:der_noi_est} is independent of the regularization parameter $\alpha$ and of the number of Newton steps $j$. It depends only on the noise in the Fr\'echet derivative of $\Op$ at $\phidag$. In addition to this bound we have to assume a concentration inequality for the right hand side of \ref{eq:der_noi_est} which is part of Assumption \ref{ass:concentration}.

The following example interprets $\|e^{der}_{j+1}\|_{L_2}$ for the IV regression operators \eqref{eq:nonlininstreg:opeq2} and \eqref{eq:iv_quant_reg_op} in the setup of the previous example.

\begin{example}\label{example_der}
We adopt the assumptions and constructions of $\widehat \Op_{ind}$, $\widehat \Op_q$, $\widehat k_{ind}$ and $\widehat k_q$ from Example \ref{example_noi}. When Assumption \ref{ass:dif} holds, the Fr\'echet derivatives have the form
\begin{align*}
&\Ophat_{ind}'[\varphi]\psi(u,z) = \int \frac{\partial}{\partial y} \widehat k_{ind}(\varphi(x) + u,\, x,\, z) \psi(z)\, dz,\\
&\Ophat_{q}'[\varphi]\psi(z) = \int \frac{\partial}{\partial y} \widehat k_{q}(\varphi(x),\, x,\, z) \psi(z)\, dz.
\end{align*}

The Hilbert-Schmidt norm bounds the operator norm from above and is the $L^2$ norm of the integral kernels $\frac{\partial}{\partial y} \widehat k_{ind}(\varphi(x) + u,x,z)$ and $\frac{\partial}{\partial y} \widehat k_{q}(\varphi(x),x,z)$. We will denote the Hilbert-Schmidt norm by $\|\cdot\|_{HS}$. We introduce the notation $\kappa(u,x,z) := \frac{\partial}{\partial y} k_{ind}(u,x,z)$ and $\widehat\kappa_{n,h}(u,x,z) := \frac{\partial}{\partial y} \widehat k_{ind}( u,x,z)$ when a sample of size $n$ and the bandwidth $h$ are used to estimate $\widehat k_{ind}$. Accordingly, $\widehat\kappa_{1,1}$ stands for the partial derivative of the unscaled kernel.
\begin{align*}
\Ex &\left(\|e^{der}_{j+1}\|_\X^2\right) \le C_d \rho \Ex\left( \|\Tdaghat - \Tdag\|_{HS}^{2(1+\mu)}\right)\\
&= C_d \rho \Ex\left(\int \left(\widehat\kappa_{n,h}(\varphi(x) + u,x,z) - \kappa(\varphi(x) + u,x,z) \right)^2 d(u,x,z)\right)^{1+\mu}\\
&= C_d \rho \Ex\bigg(\int \left(\widehat\kappa_{n,h}(\varphi(x) + u,x,z)-\Ex\widehat\kappa_{n,h}(\varphi(x) + u,x,z)\right)^2\\
&\qquad + \left(\Ex\widehat\kappa_{n,h}(\varphi(x) + u,x,z) - \kappa(\varphi(x) + u,x,z)\right)^2  d(u,x,z)\bigg)^{1+\mu}\\
&\le 2^{1+\mu} C_d \rho \int \Ex \big|\widehat\kappa_{n,h}(\varphi(x) + u,x,z)-\Ex\widehat\kappa_{n,h}(\varphi(x) + u,x,z)\big|^{2(1+\mu)} d(u,x,z)\\
&\qquad + 2^{1+\mu} C_d \rho\left(\int \big(\Ex\widehat\kappa_{n,h}(\varphi(x) + u,x,z) - \kappa(\varphi(x) + u,x,z)\big)^2  d(u,x,z)\right)^{1+\mu}~.
\end{align*}
Jensen's inequality was used in the last inequality. We analyze the second term first. Here 
$\Ex\, \widehat\kappa_{n,h}(\varphi(x) + u,x,z) - \kappa(\varphi(x) + u,x,z)$ is the bias of a partial derivative of a $1+d_X+d_Z$-dimensional kernel density estimator. Hence,
\[
\left(\int\big(\Ex\, \widehat\kappa_{n,h}(\varphi(x) + u,x,z) - \kappa(\varphi(x) + u,x,z)\big)^2  d(u,x,z)\right)^{1+\mu} = \Or\left(h^{2(r-1)(1+\mu)}\right).
\]
The expectation in the first term can be analyzed with the usual change in variables 
\begin{align*}
\Ex &\big|\widehat\kappa_{n,h}(\varphi(x) + u,x,z)-\Ex\widehat\kappa_{n,h}(\varphi(x) + u,x,z)\big|^{2(1+\mu)}\\
&= \int  \big|\widehat\kappa_{n,h}(\varphi(x) + u,x,z)-\Ex\widehat\kappa_{n,h}(\varphi(x) + u,x,z)\big|^{2(1+\mu)} f_{YXZ}(\tilde y, \tilde x, \tilde z) d(\tilde y,\tilde x,\tilde z)\\
&= \frac{h^{(d_X+d_Z+1)}}{n^{1+\mu}h^{2(1+\mu)(d_X+d_Z+2)}}\int  \big|\widehat\kappa_{1,1}(\bar u,\bar x,\bar z)-\Ex\widehat\kappa_{1,1}(\bar u,\bar x,\bar z)\big|^{2(1+\mu)}\\
&\hspace{140pt} f_{YXZ}(y -h(\varphi(x) + \bar u),x+ h\bar x, z+ h\bar z) d(\bar y,\bar x,\bar z)\\
&=n^{-1-\mu}h^{-((1+2\mu)(d_X+d_Z+2)+1)} \big( C f_{YXZ}(y,x,z) + \Or(h)\big) + \Or(n^{-1-\mu}). 
\end{align*}
The constant $C$ in the last line does not depend on $n$ or $h$. Combining the analysis of both terms yields
\[
\Ex \left(\|e^{der}_{j+1}\|_\X^2\right) = \Or\left(n^{-1-\mu} h^{-((1+2\mu)(d_X+d_Z+2)+1)} + h^{2(r-1)(1+\mu)}\right).
\]
A similar computation can be carried out for the quantile regression problem with operator \eqref{eq:iv_quant_reg_op}. 
\end{example}


\subsubsection{Nonlinearity error}


A restriction on the nonlinearity of $\Ophat_n$ is necessary to control $\|e^{nl}_{j+1}\|$ .A suitable constraint is the  Lipschitz condition \eqref{eq:lipschitz} in Assumption \ref{ass:lipschitz}. It  allows to bound the Taylor reminder of the first term in the nonlinearity error by
\[
\|\Ophat_n(\phidag) - \Ophat_n(\phihat_j) + \That_{n,j}(\phihat_j - \phidag)\|_\Y \leq \frac {L} {2} \|\phihat_j - \phidag\|_\X^2 = \frac {L} {2} \|e_j\|_\X^2.
\]
For the norm of the second term in $e^{nl}_{j+1}$ an additional inequality is needed. It was shown in \cite{BK:04b} Chapter 4.1 that for every $\mu \geq \frac 1 2$ there is a constant $ C_\mu$, such that for two linear operators \mbox{$A, B: \X \rightarrow \Y$} between Hilbert spaces $\|[r_\alpha(A^*A) -r_\alpha(B^*B)](B^*B)^\mu\|_{\calL(\X,\Y)} \leq C_\mu \|A-B\|_{\calL(\X,\Y)}$.
This yields in our case 
\begin{align*}
\|[r_{\alpha_j}(\That_{n,j}^*\That_{n,j}) - r_{\alpha_j}(\Tdaghat^* \Tdaghat)]\Lambda(\Tdaghat^* \Tdaghat) \omega\|_{\calL(\X,\X)} &\leq C_\mu\|\That_{n,j} - \Tdaghat\|_{\calL(\X,\Y)}\rho\\
 &\leq C_\mu \rho L \|\phihat_j - \phidag\|_\X = C_\mu \rho L \|e_j\|_\X.
\end{align*}
Putting both estimates together and use \eqref{eq:g(TT)T} gives
\begin{equation}\label{eq:enl}
\|e^{nl}_{j+1}\|_\X \leq \frac {L\sqrt{m}}{2\sqrt{\alpha_j}} \|e_j\|_\X^2 + C_\mu \rho L \|e_j\|_\X.
\end{equation}

The next Lemma computes an appropriate stopping parameter $J_{max}$ such that $\|e^{nl}_{j}\|$ is dominated by the other error components for all $j \le J_{max}$.

\begin{lemma}\label{lem:nonlin}
Let Assumptions \ref{ass:dif}, \ref{ass:source}, \ref{ass:lipschitz}, \ref{ass:saturation} hold true with a sufficiently small $\rho$ in Assumption \ref{ass:source}. Assume that $B_{2R}(\varphi_0) \subset \dom(\Op)$ and that $\phidag \in B_R(\varphi_0)$. Choose a monotonically increasing function $\Phi$ such that $\|e^{noi}_j + e^{der}_j\|_\X \leq \Phi(j)$ for all $j \ge 0$. Define
\begin{align}\label{eq:Jmax}
J_{max} := \max \left\{ j\in \N : \frac{\Phi(j)}{\sqrt{\alpha_j}} \leq C_{stop} \right\} \text{ with } 0 < C_{stop} \leq min \left\{\frac{1}{8L\sqrt{m}} , \frac{R}{4\sqrt{\alpha_0}}\right\}.
\end{align}
Then it holds for all $j:= 1,\; 2, \ldots, J_{max}$ that 
\[
\|e^{nl}_j\|_\X \leq \gamma_{nl} \left(\|e^{app}_j\|_\X + \Phi(j)\right) \text{ and } \phihat_j \in B_R(\phidag),
\]
with $\gamma_{nl} := 8L \sqrt{m} C_{stop} \leq 1$.
\end{lemma}

\begin{proof}

We generalize the proof strategy of Lemma 2.2 in \cite{BauHohMun:09} to our setting. The proposition follows by induction in $j$. We start with the induction step. Assume that the proposition holds for $j-1$ with $2 \leq j \leq J_{max}$. Since $\Phi$ is increasing and by \eqref{eq:gammaapp}
\begin{align*}
\|e_{j-1}\| &\leq (1+\gamma_{nl})\left(\|e^{app}_{j-1}\| + \Phi(j-1)\right)\\
&\leq (1+\gamma_{nl})\left(\gamma_{app}\|e^{app}_j\| + \Phi(j)\right).
\end{align*}
Combining this with inequality \eqref{eq:enl} and using $(a+b)^2\leq2a^2+2b^2$ yields
\begin{equation}\label{eq:ugl1}
\begin{split}
 \|e^{nl}_j\| \leq \; & C_\mu \rho L (1+\gamma_{nl})\left(\gamma_{app}\|e^{app}_j\| + \Phi(j)\right)\\
& + \frac {L \sqrt{m}} {\sqrt{\alpha_j}} (1+\gamma_{nl})^2\left(\gamma_{app}^2\|e^{app}_j\|^2 + \Phi(j)^2\right).
\end{split}
\end{equation}

If $\rho \leq \gamma_{nl} / (2 C_\mu (1+\gamma_{nl})\gamma_{app})$, the first line on the right hand side is bounded by $1/2\gamma_{nl}\left(\|e^{app}_j\| + \Phi(j)\right)$. To bound the second line, we assume that $\rho \leq \gamma_{nl} / (2 C_\Lambda \alpha_0^{\mu- 1/2} L \sqrt{m} (1+\gamma_{nl})^2 \gamma_{app}^2)$. It follows from \eqref{eq:eapp} that
\[
\frac {\|e^{app}_j\|}{\sqrt{\alpha_j}} \leq C_\Lambda \rho \alpha_j^{\mu- \frac 1 2 } \leq C_\Lambda \rho \alpha_0^{\mu- \frac 1 2 } \leq \frac{\gamma_{nl}} {2 L (1+\gamma_{nl})^2 \gamma_{app}^2} \; .
\]
Thus, $L / \sqrt{\alpha_j} (1+\gamma_{nl})^2 \gamma_{app}^2\|e^{app}_j\|^2 \leq \frac 1 2 \gamma_{nl}\|e^{app}_j\|$. 
By the definition of $J_{max}$ the fact that $\gamma_{nl} \leq 1$ we have
\[
\frac {L\sqrt{m}} {\sqrt{\alpha_j}} (1+\gamma_{nl})^2 \Phi^2(j) \leq \frac {4 L\sqrt{m}} {\sqrt{\alpha_j}} \Phi^2(j) \leq 4 L \sqrt{m} C_{stop} \Phi(j) \leq \frac {\gamma_{nl}}{2} \Phi(j).
\]
Therefore, the second line on the right hand side of \eqref{eq:ugl1} is also bounded by $\frac{1}{2}\gamma_{nl}(\|e^{app}_j\| + \Phi(j))$. Together with the estimation of the first line this gives
\[
\|e^{nl}_j\| \leq \gamma_{nl} \left(\|e^{app}_j\| + \Phi(j)\right).
\]

The base case $j=1$ of the induction follows in exactly the same way, as long as $\alpha_0$ is large enough which is guaranteed by Assumption \ref{ass:saturation} and \eqref{eq:gammaapp}.\medskip

Finally, we have to show that $\phihat_j \in B_R(\phidag)$. If $\rho \leq R / \left(2 C_\Lambda \alpha_0^\mu (1+\gamma_{nl})\right)$, then
\[
\|e^{app}_j\| \leq C_\Lambda \rho \alpha_j^\mu \leq C_\Lambda \rho \alpha_0^\mu \leq \frac R {2(1+\gamma_{nl})}.
\]
Moreover, the monotonicity of $\Phi$ and the definitions of $J_{max}, \; C_{stop}$ and $\gamma_{nl}$ imply:
\[
\Phi(j) \leq \Phi(J_{max}) \leq C_{stop}\sqrt{\alpha_{J_{max}}} \leq C_{stop}\sqrt{\alpha_0} \leq \frac R 4 \leq \frac R {2(1+\gamma_{nl})} \; .
\]
This shows together with the first part of the proof that 
\[
\|e_j\| \leq (1+\gamma_{nl})\left(\|e^{app}_j\| + \Phi(j)\right) \leq R.
\] 
Hence,
$\phihat_j \in B_R(\phidag) \subset \dom(\Op)$.
\end{proof}

The assumption that $\rho$ is sufficiently small means that the initial guess must be close enough to the true solution. As always for Newton type methods we get only local convergence. In practice, the convergence radius seems to be quite large and does usually not restrict the applicability of the method.

%

\subsection{Convergence rates with a priori parameter choice}

This section presents the proof to Theorem \ref{theo:spec_conv}. We generalize the proof strategy in \cite{BauHohMun:09} to our setting. We start with a lemma about deterministic errors, i.e. $0=\Var(\|\Ophat(\phidag)\|_\Y)= \Var(\|\Tdaghat\|_{\calL(\X,\Y)})$. The crucial point is to show that the maximal stopping parameter $J_{max}$ from in Lemma \ref{lem:nonlin} is larger or equal to a suitable stopping parameter.

\begin{lemma}\label{lem:deterministic}
Suppose that  the Assumptions \ref{ass:dif}, \ref{ass:source}, \ref{ass:saturation}, and \ref{ass:lipschitz} are fulfilled. Assume that $B_{2R}(\varphi_0) \subset \dom(\Op)$, and that $\rho$ is small enough as in Lemma \ref{lem:nonlin}. Let $\widetilde \delta^{noi}_n$ and $\widetilde \delta^{der}_n$ be a sequence such that $\widetilde \delta^{noi}_n \geq \|\Ophat_n(\phidag)\|_\Y$ and $\widetilde \delta^{der}_n \geq \|\Tdaghat - \Tdag\|_{\calL(\X,\Y)}^{1+\mu}$.
Set
\[
\widetilde J := \argmin\limits_{j \in \N}\left(\|e^{app}_j\|_\X + \sqrt{\frac{m}{\alpha_j}} \widetilde \delta^{noi}_n \right) \quad \text{and} \quad J := \min \{J_{max}, \widetilde J\}.
\]
Then there exists a constant $C$ 
such that
\[
\|\phihat_{J} - \phidag\|_\X \leq C \inf\limits_{j\in\N}\left(\|e^{app}_j\|_\X + \sqrt{\frac{m}{\alpha_j}} \widetilde \delta^{noi}_n + C_d \rho \widetilde \delta^{der}_n\right). 
\]
\end{lemma}

\begin{proof}
Notice that $J$ also minimizes 
\[
\argmin\limits_{j \in \N} \left(\|e^{app}_j\| + \sqrt{\frac{m}{\alpha_j}} \widetilde \delta^{noi}_n + C_d \rho \widetilde \delta^{der}_n\right)
\] 
because $C_d \rho \widetilde \delta^{der}_n$ does not depend on $j$. Set $\Phi(j) := \sqrt{m/\alpha_j} \widetilde \delta^{noi}_n + C_d \rho \widetilde \delta^{der}_n$. If $\widetilde J \leq J_{max}$, the theorem is proven by Lemma \ref{lem:nonlin} with $C = 1 + \gamma_{nl}$.

If $\widetilde J > J_{max}$, then $\widetilde J \geq J_{max}+1$ and $\Phi(J_{max}+1) / C_{stop} \geq \sqrt{\alpha_{J_{max}+1}}$. Hence, by the monotonicity of $\Phi$
\begin{align*}
\left(1 + \frac{C_\Lambda \rho \alpha_0^{\mu-\frac1 2}} {C_{stop} \sqrt{q_\alpha}} \right) \left(\|e^{app}_J\| + \Phi(\widetilde J)\right)
&\geq \left(1 + \frac{C_\Lambda \rho \alpha_0^{\mu-\frac 1 2}} {C_{stop} \sqrt{q_\alpha}} \right) \Phi(J_{max}+1)\\
&\geq \Phi(J_{max}) + C_\Lambda \rho \frac{\Phi(J_{max}+1) \alpha_0^{\mu-\frac1 2}} {C_{stop} \sqrt{q_\alpha}}\\
&\geq \Phi(J_{max}) + C_\Lambda \rho \frac{\sqrt{\alpha_{J_{max}+1}} \alpha_0^{\mu-\frac1 2}} {\sqrt{q_\alpha} }\\
&= \Phi(J_{max}) + C_\Lambda \rho \sqrt{\alpha_{J_{max}}} \alpha_0^{\mu-\frac 1 2}\\
&\geq \Phi(J_{max}) + C_\Lambda \rho \alpha_{J_{max}}^\mu\\
&\geq \Phi(J_{max}) + \|e^{app}_{J_{max}}\|.
\end{align*}

This proves the lemma when $\widetilde J > J_{max}$ with 
\[
C = \left(1 + \frac{C_\Lambda \rho \alpha_0^{\mu-\frac1 2}} {C_{stop} \sqrt{q_\alpha}} \right)(1+ \gamma_{nl}).
\]
\end{proof}

This lemma implies convergence in probability of the estimator with the same rate. We can now proof Theorem \ref{theo:spec_conv}.

\begin{proof}(of Theorem \ref{theo:spec_conv})
We introduce the notation
\begin{align*}
&\delta^{noi}_n = \Ex\big[\|\Ophat_n(\phidag)\|_\Y\big], &&(\sigma^{noi}_n)^2 = \Var\big(\|\Ophat_n(\phidag)\|_\Y\big),\\
&\delta^{der}_n = \Ex\big[\|\Tdaghat - \Tdag\|_{\calL(\X,\Y)}^{1+\mu}\big], &&(\sigma^{der}_n)^2 = \Var\big(\|\Tdaghat - \Tdag\|_{\calL(\X,\Y)}^{1+\mu}\big).
\end{align*}
Similar to the last proof $J$ is also a minimizer of 
\[
J = \argmin\limits_{j \in \N} \left(\|e^{app}_j\| + \sqrt{\frac{m}{\alpha_j}}(\delta^{noi}_n + \sigma^{noi}_n) + C_d \rho (\delta^{der}_n + \sigma^{der}_n)\right).
\]
The proof uses a threshold argument. The key tool is the following construction. Define a chain of events with increasing noise level containing each other as $A_1 \subset A_2 \subset \ldots \subset A_{k_{max}}$ by
\begin{align}\label{eq:def_A}
A_k := \left\{ \phihat_j \in B_{2R}(\varphi_0) \text{ and } \|e^{noi}_j + e^{der}_j\| \leq \Phi^{noi}_n(\tau_k, j) \text{ for all } j = 1, \ldots, J \right\}
\end{align}
and
\[
k_{max} := \max \left\{\left\lfloor\frac{\ln\big((\sigma^{noi}_n)^{-2}\big)}{c_2}\right\rfloor , \left\lfloor\frac{\ln\left((\sigma^{der}_n)^{-2}\right)}{c_4}\right\rfloor \right\}
\]
with $c_2$ and $c_4$ from \eqref{eq:concentration_noi} and \eqref{eq:concentration_der}, and with
\begin{align}\label{eq:Phi}
&\Phi^{noi}_n(\tau , j) := \sqrt{\frac{m}{\alpha_j}}\delta^{noi}_n + C_d \rho \delta^{der}_n + \sqrt{\tau(j)}\left(\sqrt{\frac{m}{\alpha_j}}\sigma_n^{noi} + C_d \rho \sigma^{der}_n\right).
\end{align}
Set $\tau_k(j) := k + c_2^{-1}\ln(\kappa)(J - j)$ with some $\kappa > 1$ small enough such that 
\begin{equation}\label{eq:tau}
\tau(j+1) q_\alpha \geq \tau(j)
\end{equation}
with $q_\alpha$ as in \eqref{eq:alpha} for all $j$. Consequently, $\Phi^{noi}_n(\tau_k, j)$ is monotonically increasing in $j$ as required for the application of Lemma \ref{lem:nonlin}. Notice that $k_{max}$ is chosen in a way such that 
\[
\max\left\{e^{-c_2k_{max}} , e^{-c_4k_{max}}\right\} \leq \ \max\{(\sigma^{noi}_n)^2 , (\sigma^{der}_n)^2\}.
\]

Lemma \ref{lem:nonlin} and Lemma \ref{lem:2implies1} below show that $\|e^{noi}_j + e^{der}_j\| \leq \Phi^{noi}_n(\tau_k, j)$ implies $\phihat_j \in B_{2R}(\varphi_0)$ when $\sigma^{noi}_n$ is sufficiently small, i.e. the second condition in the definition of $A_k$ implies the first one.


In order to prepare the final step of the proof, we bound the probability of $A_k \backslash A_{k-1}$ and the probability of the event complementary to $A_k$. The following computation uses \eqref{eq:lipschitz}, \eqref{eq:concentration_noi}, \eqref{eq:concentration_der}.
\begin{align*}
P(A_k \backslash A_{k-1}) &= P \left\{\Phi^{noi}_n(\tau_{k-1}, j) < \|e^{noi}_j + e^{der}_j\| \leq \Phi^{noi}_n(\tau_k, j) \text{ for all } j = 1, \ldots, J\right\}\\
&\leq P\left\{\Phi^{noi}_n(\tau_{k-1}, j) < \|e^{noi}_j + e^{der}_j\| \text{ for all } j = 1, \ldots, J \right\}\\
&\leq \sum\limits_{j=1}^{J} c_1 e^{-c_2\tau_k(j)} + c_3 e^{-c_4\tau_k(j)} \leq (c_1 e^{-c_2 k} + c_3 e^{-c_4 k}) \sum\limits_{j=1}^{J} \kappa^{j-J}\\
&\leq (c_1 e^{-c_2 k} + c_3 e^{-c_4 k}) \sum\limits_{j=0}^\infty \kappa^{j} = \frac {c_1 e^{-c_2 k} + c_3 e^{-c_4 k}} {1-\kappa^{-1}}\\
P(\mathcal{C}A_k) &\leq P\left\{\Phi^{noi}_n(\tau_{k-1}, j) < \|e^{noi}_j + e^{der}_j\| \text{ for all } j = 1, \ldots, J \right\}\\
&\leq (c_1 e^{-c_2 k} + c_3 e^{-c_4 k}) \sum\limits_{j=0}^\infty \kappa^{j} = \frac {c_1 e^{-c_2 k} + c_3 e^{-c_4 k}} {1-\kappa^{-1}} \; .
\end{align*}
In every event $A_k$ we have $J = J^*$. The assumptions of Lemma \ref{lem:nonlin} are fulfilled in $A_k$. This implies the following error bound
\begin{align*}
\|\phihat_{J}& - \phidag\|^2 \leq \bigg[\|e^{app}_{J}\| + \sqrt{\frac{m}{\alpha_{J}}} \delta^{noi}_n + C_d \rho \delta^{der}_n   + \sqrt{\tau_k(J)}\left( \sqrt{\frac{m}{\alpha_{J}}}\sigma^{noi}_n + C_d \rho \sigma^{der}_n\right)\bigg]^2\\
&\leq 10\|e^{app}_J\|^2 + 10\frac{m}{\alpha_{J}} (\delta^{noi}_n)^2 + 10(\delta^{der}_n)^2 C_d^2 \rho^2   + 10k \frac{m}{\alpha_{J}} (\sigma^{noi}_n)^2 + 10kC_d^2 \rho^2(\sigma^{der}_n)^2\\
&=: C_k.
\end{align*}
By the construction of the algorithm \eqref{eq:it_tik} the worst case error is $\|\phihat_{J^*} -\phidag\| \leq 3R$. This will serve as an error bound in the event $\mathcal{C}A_{k_{max}}$. Putting everything together yields
\begin{align*}
\Ex(\|&\phihat_{J^*} -\phidag\|^2) \leq  P(A_1)C_1 + \sum\limits_{k=2}^{k_{max}}P(A_k\backslash A_{k-1})C_k + P(\mathcal{C}A_{k_{max}})9R^2\\
&\leq 10\left( \|e^{app}_{J}\|^2 + \frac{m}{\alpha_{J}} (\delta^{noi}_n)^2  + (\delta^{der}_n)^2 C_d^2 \rho^2\right)\\
&\quad + 10 P(A_1)\left(\frac{m}{\alpha_{J}} (\sigma^{noi}_n)^2 + (\sigma^{der}_n)^2 C_d^2 \rho^2 \right)\\
& \quad + \sum\limits_{k=2}^{k_{max}} P(A_k \backslash A_{k-1}) \left(10k\frac{m}{\alpha_{J}} (\sigma^{noi}_n)^2 + 10k(\sigma^{der}_n)^2 C_d^2 \rho^2 \right) + P(\mathcal{C}A_{k_{max}})9R^2\\[6pt]
&\leq 10\left( \|e^{app}_{J}\|^2 + \frac{m}{\alpha_{J}} (\delta^{noi}_n)^2 + (\delta^{der}_n)^2 C_d^2 \rho^2 \right) + P(\mathcal{C}A_{k_{max}})9R^2\\
& \quad + 10\left(\frac{m}{\alpha_{J}} (\sigma^{noi}_n)^2 + (\sigma^{der}_n)^2 C_d^2 \rho^2 \right) \left(2 + \sum\limits_{k=3}^{k_{max}} k P(A_k \backslash A_{k-1}) \right) \\[6pt]
&\leq 10\left( \|e^{app}_{J}\|^2 + \frac{m}{\alpha_{J}} (\delta^{noi}_n)^2 + (\delta^{der}_n)^2 C_d^2 \rho^2 \right) + \left( \frac {c_1 e^{-c_2 k_{max}} + c_3 e^{-c_4 k_{max}}} {1-\kappa^{-1}}\right) 9R^2\\
& \quad + 10\left( \frac{m}{\alpha_{J}} (\sigma^{noi}_n)^2 + (\sigma^{der}_n)^2 C_d^2 \rho^2 \right) \left(2 + \sum\limits_{k=2}^{k_{max}-1} (k+1) \frac {c_1 e^{-c_2 k} + c_3 e^{-c_4 k}} {1-\kappa^{-1}} \right) \\[6pt]
&\leq 10\left( \|e^{app}_{J}\|^2 + \frac{m}{\alpha_{J}} (\delta^{noi}_n)^2 + (\delta^{der}_n)^2 C_d^2 \rho^2 \right) + (c'\max\{(\sigma^{noi}_n)^2, (\sigma^{der}_n)^2\}) 9R^2\\
& \quad + 10\left(\frac{m}{\alpha_{J}} (\sigma^{noi}_n)^2 + (\sigma^{der}_n)^2 C_d^2 \rho^2 \right) \left(2 + \sum\limits_{k=2}^{\infty} (k+1) \frac {c_1 e^{-c_2 k} + c_3 e^{-c_4 k}} {1-\kappa^{-1}} \right) \\[6pt]
&\leq 10 \left( \|e^{app}_{J}\|^2 + \frac{m}{\alpha_{J}} (\delta^{noi}_n)^2 + (\delta^{der}_n)^2 C_d^2 \rho^2 \right) + (c'\max\{(\sigma^{noi}_n)^2, (\sigma^{der}_n)^2\}) 9R^2\\
& \quad + 10 c'' \left(\frac{m}{\alpha_{J}} (\sigma^{noi}_n)^2 + (\sigma^{der}_n)^2 C_d^2 \rho^2 \right)\\[6pt]
&\leq C \left(\|e^{app}_{J}\|^2 + \frac{m}{\alpha_{J}} \left[(\delta^{noi}_n)^2 +(\sigma^{noi}_n)^2 \right] + C_d^2 \rho^2\left[(\delta^{der}_n)^2  + (\sigma^{der}_n)^2\right]\right)\\
&= C \left(\|e^{app}_{J}\|^2 + \frac{m}{\alpha_{J}} \Ex\big[\|\Ophat_n(\phidag)\|^2_\Y\big] + C_d^2 \rho^2\Ex\big[\|\Tdaghat - \Tdag\|_{\calL(\X,\Y)}^{2+2\mu}\big]\right)\\
&=\Or \left(\left(\Ex(\|\Ophat_n(\phidag)\|_\Y^2)\right)^{\frac{2\mu}{2\mu + 1}} + \Ex\left(\|\Tdaghat - \Tdag\|_{\calL(\X,\Y)}^{2+2\mu}\right) \right)~.
\end{align*}
We used $P(A_1) + \sum\limits_{k=2}^{k_{max}}P(A_k\backslash A_{k-1}) + P(\mathcal{C}A_{k_{max}}) = 1$ and $P(A_1) + P(A_2 \backslash A_1) \leq 1$. Furthermore, $c'> 0$, $c''>0$ and $C>0$ are generic constants.
\end{proof}

The following two lemmas are needed for the proof of Theorem \ref{theo:spec_conv} above.

\begin{lemma}
Let the assumptions of Theorem \ref{theo:spec_conv} hold and define:
\begin{align*}
\tilde\Phi(j) &:=  \sqrt{\frac{m}{\alpha_j}}(\delta^{noi}_n + \sigma^{noi}_n) + C_d \rho (\delta^{der}_n + \sigma^{der}_n)\\
\underline\Gamma_{noi} &:= \frac{\sqrt{m / (q_\alpha \alpha_1)}(\delta^{noi}_n + \sigma^{noi}_n) + C_d \rho (\delta^{der}_n + \sigma^{der}_n)}{\sqrt{m / \alpha_1}(\delta^{noi}_n + \sigma^{noi}_n) + C_d \rho (\delta^{der}_n + \sigma^{der}_n)}\\[6pt]
\overline\Gamma_{noi} &:= q_\alpha^{-\frac{1}{2}}.
\end{align*}

The following two bounds hold for the stopping index $J$ in Theorem \ref{theo:spec_conv}:
\begin{eqnarray}
\label{l1f1}&&(1-\underline\Gamma_{noi}^{-1}) \tilde \Phi(J) \leq (\gamma_{app}-1) \|e_J^{app}\|,\\
\label{l1f2}&&J \geq 
\sup\left\{k \in \N \bigg\vert \|e^{app}_1\|\gamma_{app}^{1-k} > \inf_{l\in\N}\left(C_\Lambda\rho\sqrt{\alpha_l} + \tilde\Phi(1)\overline\Gamma_{noi}^{l-1} \right) \right\}.
\end{eqnarray}
\end{lemma}

\begin{proof}
Note that \eqref{eq:gammaapp} implies
\begin{equation}\label{l1p1}
1 < \underline\Gamma_{noi} \leq \frac{\tilde\Phi(j+1)}{\tilde\Phi(j)} \leq \overline\Gamma_{noi}, \quad \text{for all } j \in \N.
\end{equation}

We start with inequality \eqref{l1f1}. Assume the opposite holds true
\[
(1-\underline\Gamma_{noi}^{-1}) \tilde \Phi^{noi}_n(J) > (\gamma_{app}-1) \|e_{J}^{app}\|.
\]
It would follow from \eqref{eq:gammaapp} and \eqref{l1p1} that
\[
\|e^{app}_{J-1}\| + \tilde \Phi(J - 1) \leq \gamma_{app}\|e^{app}_{J}\| + \underline\Gamma_{noi}^{-1} \tilde \Phi(J) < \|e^{app}_{J}\| + \tilde \Phi(J).
\] 
This is a contradiction to the definition of $J$ and therefore proves \eqref{l1f1}.\medskip

In order to prove \eqref{l1f2} assume that for some $k$, and some $l \geq 1$ 
\[
\|e^{app}_1\|\gamma_{app}^{1-k} > C_\Lambda\rho\sqrt{\alpha_l} + \tilde\Phi(1)\overline\Gamma_{noi}^{l-1}.
\]
It follows from \eqref{eq:eapp}, \eqref{eq:gammaapp} and \eqref{l1p1} that for all $j \leq k$
\begin{equation*}
\begin{split}
\|e^{app}_l\| + \tilde \Phi(l) &\leq C_\Lambda \rho \sqrt{\alpha_l} + \tilde \Phi(1)\overline\Gamma_{noi}^{l-1} < \|e^{app}_1\|\gamma_{app}^{1-k} \leq \|e^{app}_k\| \leq \|e^{app}_j\|\\
&\leq \|e^{app}_j\| + \tilde \Phi(j).
\end{split}
\end{equation*}
As $J$ is the minimizer for $\|e^{app}_j\| + \tilde \Phi(j)$ this implies $J >k$. Taking the infimum over $l$ and the supremum over $k$ proves the lemma.
\end{proof}

\begin{lemma}\label{lem:2implies1}
Let the assumptions of Theorem \ref{theo:spec_conv} hold true. Define $J_{max}$ as in Lemma \ref{lem:nonlin} with $\tau_k(j) := k + \frac{\ln(\kappa)}{c_2}(J - j)$ 
\begin{align*}
J_{max}(k) := \max \bigg\{j \in \N \bigg\vert \bigg[&\sqrt{\frac{m}{\alpha_j}}\delta^{noi}_n  + C_d \rho \delta^{der}_n\\
&+ \sqrt{\tau_k(j)} \left(\sqrt{\frac{m}{\alpha_j}} \sigma^{noi}_n + C_d \rho \sigma^{der}_n\right) \bigg] \alpha_j^{-\frac{1}{2}} \leq C_{stop} \bigg\}.
\end{align*}
There exist $\bar\sigma^{noi}_n > 0$ and $\bar\sigma^{der}_n > 0$ such that for all $\sigma^{noi}_n \leq \bar\sigma^{noi}_n$ and $\sigma^{der}_n \leq \bar\sigma^{der}_n$ and for all \hbox{$k = 1, \ldots, k_{max}$} it holds that $J \leq J_{max}$.
\end{lemma}

\begin{proof}
Since $\tau_k(j)$ fulfills inequality \eqref{eq:tau} for $k\leq k_{max}$ and $j \leq J$, 
\[
\tau_k(J) \leq \tau_{k_{max}}(J) \leq \max \left\{ \ln((\sigma^{noi}_n)^{-2}) / c_2, \, \ln((\sigma^{der}_n)^{-2}) / c_4 \right\}.
\]
Hence,
\begin{equation*}
\begin{split}
&\left(\sqrt{\frac{m}{\alpha_j}} \delta^{noi}_n  + C_d \rho \delta^{der}_n + \sqrt{\tau_k(j)} \left(\sqrt{\frac{m}{\alpha_j}} \sigma^{noi}_n + C_d \rho \sigma^{der}_n \right) \right) \alpha_j^{-\frac{1}{2}}\\
&\qquad \qquad \leq \left( \sqrt{\frac{m}{\alpha_{J}}}\delta^{noi}_n + C_d \rho \delta^{der}_n + \sqrt{\tau_{k_{max}}(J)} \left( \sqrt{\frac{m}{\alpha_{J}}}\sigma^{noi}_n + C_d \rho \sigma^{der}_n \right) \right) \alpha_{J}^{-\frac{1}{2}}\\
&\qquad  \qquad \leq \max \left\{ \sqrt{\frac{\ln((\sigma^{noi}_n)^{-2})}{c_2}}, \sqrt{\frac{\ln((\sigma^{der}_n)^{-2})}{c_4}}\right\} \tilde\Phi(J)\alpha_J^{-\frac{1}{2}}\\
&\qquad  \qquad \leq \max \left\{ \sqrt{\frac{\ln((\sigma^{noi}_n)^{-2})}{c_2}}, \sqrt{\frac{\ln((\sigma^{der}_n)^{-2})}{c_4}}\right\} \frac{\gamma_{app}-1}{1-\underline\Gamma_{noi}^{-1}} \|e^{app}_{J}\|\alpha_{J}^{-\frac{1}{2}}\\
&\qquad  \qquad \leq C \max \left\{ \sqrt{\ln((\sigma^{noi}_n)^{-2})} , \sqrt{\ln((\sigma^{der}_n)^{-2})} \right\} \alpha_{J}^{\mu-\frac{1}{2}}
\end{split}
\end{equation*}
with $\displaystyle C:=\frac{\rho C_\Lambda (\gamma_{app} - 1)}{\min\{c_2 , c_4\} (1-\underline\Gamma_{noi}^{-1})}$.\medskip

Moreover, we have to take into account that in inequality \eqref{l1f2}
\begin{equation*}
\begin{split}
&\inf_{l\in\N}\left(C_\Lambda\rho\sqrt{\alpha_l} + \tilde\Phi(1)\overline\Gamma_{noi}^{l-1} \right)\\
&\qquad = \inf_{l\in\N}\left(C_\Lambda\rho\sqrt{\alpha_l} + \left(\alpha_1^{-\frac{1}{2}}(\delta^{noi}_n + \sigma^{noi}_n) + C_d \rho (\delta^{der}_n + \sigma^{der}_n) \right)\overline\Gamma_{noi}^{l-1} \right)
\end{split}
\end{equation*}
decays with a polynomial rate in $\sigma^{noi}_n$ and $\sigma^{der}_n$. Therefore, there exists a constant $b$ for which $J \geq - b \max\{\ln(\sigma^{noi}_n), \ln(\sigma^{der}_n)\}$, while $\displaystyle\lim_{x\rightarrow \infty}xq_\alpha^{cx}$ goes to $0$ for every $c$ as $q_\alpha<1$. Hence, there are $\bar\sigma^{noi}_n$ and $\bar\sigma^{der}_n$ such that for all $\sigma^{noi}_n \in \; ]0, \bar\sigma^{noi}_n]$ and for all $\sigma^{der}_n \in \; ]0, \bar\sigma^{der}_n]$ it holds:
\begin{equation*}
\begin{split}
&C \max \left\{ \sqrt{\ln((\sigma^{noi}_n)^{-2})} , \sqrt{\ln((\sigma^{der}_n)^{-2})} \right\} \alpha_{J}^{\mu-\frac{1}{2}}\\
& \leq C \max \left\{ \sqrt{\ln((\sigma^{noi}_n)^{-2})} , \sqrt{\ln((\sigma^{der}_n)^{-2})} \right\} \left(\frac{\alpha_0}{q_\alpha}\right) q_\alpha^{\max \left\{ \sqrt{\ln((\sigma^{noi}_n)^{-2})} , \sqrt{\ln((\sigma^{der}_n)^{-2})} \right\} \frac{b}{2}(\mu-\frac{1}{2})}\\
&\leq C_{stop}.
\end{split}
\end{equation*}
Together with the first estimate this proves the lemma.
\end{proof}

Finally, we can proof Corallary \ref{cor:rates_npiv}.
\begin{proof}(of Corollary \ref{cor:rates_npiv})

Combining the results of Theorem \ref{theo:spec_conv} and Examples \ref{example_noi} and \ref{example_der} we get the rate
\begin{align*}
\Ex(\|&\phihat_{J^*} -\phidag\|_\X^2)\\ 
&= \Or \left(\rho^{\frac{2}{2\mu +1}}(n^{-1} h^{-(d_Z+1)})^{\frac{2\mu}{2\mu + 1}} + n^{-1-\mu} h^{-((1+2\mu)(d_X+d_Z+2)+1)} + h^{\frac{4\mu r}{2\mu + 1}} \right).
\end{align*}
\end{proof}

\subsection{Convergence rates for adaptive estimation}

%

\begin{proof}(of Theorem \ref{theo:lepskii_risk})
When $\widetilde\Phi^{noi}_n$ is used in the definition of $J_{max}$ in \eqref{eq:Jmax}, it follows that $J_{max} = \Or\big[\ln((\sigma^{noi}_n)^{-1}) + \ln((\sigma^{der}_n)^{-1})\big]$. Consider the event $A$ defined as in \eqref{eq:def_A} with
\[
\tau(j) : = \max \left\{\frac{\ln\big((\sigma^{noi}_n)^{-2}\big)}{c_2},\, \frac{\ln\left((\sigma^{der}_n)^{-2}\right)}{c_4} \right\}.
\]
Applying the Lepski{\u\i} principle (e.g. Corollary 1 in \cite{Mathe:06}) 
in this event gives the estimate
\[
\|\phihat_{Lep} - \phidag\| \leq 6 q_\alpha^{-\frac{1}{2}} (1 + \gamma_{nl}) \min_{j=1,\, \ldots,\, J_{max}} \left(\|e^{app}\| + \widetilde\Phi^{noi}_n \right).
\]
In Lemma \ref{lem:2implies1} it was shown that for sufficiently small values of $\delta^{noi}_n$, $\sigma^{noi}_n$, $\delta^{der}_n$ and $\sigma^{der}_n$ the parameter $J_{max}$ is large enough. Hence, in the asymptotics we can take the infimum over $\N$
\[
\|\phihat_{Lep} - \phidag\| \leq 6 q_\alpha^{-\frac{1}{2}} (1 + \gamma_{nl}) \inf_{j \in \N} \left(\|e^{app}\| + \widetilde\Phi^{noi}_n \right).
\]

In addition, we estimate the probability of the opposite event of $A$ by
\begin{align*}
P(\mathcal{C} A) &\leq \sum_{j=1}^{J_{max}} c_1 \exp(-\ln((\sigma^{noi}_n)^{-2}) + c_3 \exp(-\ln((\sigma^{der}_n)^{-2})\\
&\leq  J_{max} \left(c_1 (\sigma^{noi}_n)^2 + c_3 (\sigma^{der}_n)^2 \right)\\
&\leq C' \max \left\{ \ln((\sigma^{noi}_n)^{-1})(\sigma^{noi}_n)^2,\, \ln((\sigma^{der}_n)^{-1})(\sigma^{der}_n)^2 \right\}\\
&\leq C'' \min\limits_{j \in \N} \bigg( \|e^{app}_j\|^2 + \frac{m}{\alpha_j} \big[(\delta^{noi}_n)^2 + \ln((\sigma^{noi}_n)^{-1}) (\sigma^{noi}_n)^2 \big]\\
&\hspace{66mm} + C_d^2 \rho^2 \big[(\delta^{der}_n)^2 + \ln((\sigma^{der}_n)^{-1}) (\sigma^{der}_n)^2 \big] \bigg)
\end{align*}
with two constants $C'$ and $C''$. We used in the third row the fact that $J_{max} = \Or \big[\ln((\sigma^{noi}_n)^{-1}) + \ln((\sigma^{der}_n)^{-1})\big]$ and in the fourth row that $\alpha_j^{-\frac{1}{2}}$ is monotonically increasing in $j$.

We finish the proof with the estimation of the MISE
\begin{align*}
\Ex[\|\phihat_{Lep} - \phidag\|^2] &\leq P(A) 36 q_\alpha^{-1} (1 + \gamma_{nl})^2 \inf_{j \in \N} \left(\|e^{app}\| + \widetilde\Phi^{noi}_n \right)^2 + P(\mathcal{C} A) 9R^2\\
&\leq C \min\limits_{j \in \N} \bigg( \|e^{app}_j\|^2 + \frac{m}{\alpha_j} \big[(\delta^{noi}_n)^2 + \ln((\sigma^{noi}_n)^{-1}) (\sigma^{noi}_n)^2 \big]\\
&\hspace{90pt} + C_d^2 \rho^2 \big[(\delta^{der}_n)^2 + \ln((\sigma^{der}_n)^{-1}) (\sigma^{der}_n)^2 \big] \bigg).
\end{align*}
The rate in the theorem follows from the last line and the bound \eqref{eq:eapp} on $\|e_j^{app}\|$.
\end{proof}

\end{document}